\documentclass[a4paper,onecolumn,11pt,accepted=2026-01-12]{quantumarticle}
\pdfoutput=1
\usepackage[bookmarksnumbered,linktocpage,hypertexnames=false,colorlinks=false,linkcolor=blue,urlcolor=blue,citecolor=blue,anchorcolor=green,breaklinks=true,pdfusetitle]{hyperref}
\usepackage{graphicx,fullpage,xcolor}
\usepackage[shortlabels]{enumitem}
\usepackage{amsthm,amsmath,amssymb,mathtools,dsfont}
\usepackage[capitalize,nameinlink]{cleveref}

\newcommand{\ket}[1]{|#1\rangle}
\newcommand{\bra}[1]{\langle#1|}
\newcommand{\braket}[2]{\langle#1|#2\rangle}
\newcommand{\N}{\mathbb{N}}
\newcommand{\C}{\mathbb{C}}

\newcommand{\R}{\mathbb{R}}

\newcommand{\ot}{\otimes}
\newcommand{\op}{\oplus}
\newcommand{\eps}{\varepsilon}
\newcommand{\brac}[1]{\left(#1\right)}
\newcommand{\sbrac}[1]{\left[#1\right]}
\newcommand{\cbrac}[1]{\left\{#1\right\}}

\newcommand{\mcG}{\mathcal{G}}
\newcommand{\mcH}{\mathcal{H}}
\newcommand{\Id}{\mathds{1}}
\DeclareMathOperator*{\E}{\mathbb{E}}
\newcommand{\Enc}{\mathrm{Enc}}
\newcommand{\Dec}{\mathrm{Dec}}

\newtheorem{theorem}{Theorem}[section]
\newtheorem{lemma}[theorem]{Lemma}
\newtheorem{corollary}[theorem]{Corollary}

\theoremstyle{definition}
\newtheorem{definition}[theorem]{Definition}
\newtheorem{remark}[theorem]{Remark}

\numberwithin{equation}{section}

\DeclarePairedDelimiter\parens{\lparen}{\rparen}
\DeclarePairedDelimiter\norm{\lVert}{\rVert}
\DeclarePairedDelimiter\abs{\lvert}{\rvert}

\allowdisplaybreaks[4]
\begin{document}

\title{A Computational Tsirelson's Theorem for the Value of Compiled XOR Games}

\author{David Cui}
\affiliation{Department of Mathematics, Massachusetts Institute of Technology, Cambridge, MA 02139, United States}
\email{dzcui@mit.edu}

\author{Giulio Malavolta}
\affiliation{Department of Computing Sciences, Bocconi University, 20136 Milano MI, Italy}
\email{giulio.malavolta@hotmail.it}

\author{Arthur Mehta}
\affiliation{Department of Mathematics and Statistics, University of Ottawa, Ottawa, ON K1N 6N5, Canada}
\email{amehta2@uottawa.ca}

\author{Anand Natarajan}
\affiliation{Department of Electrical Engineering and Computer Science, Massachusetts Institute of Technology, Cambridge, MA 02139, United States}
\email{anandn@mit.edu}

\author{Connor Paddock}
\affiliation{Department of Mathematics and Statistics, University of Ottawa, Ottawa, ON K1N 6N5, Canada}
\affiliation{Department of Computer Science, University of Calgary, Calgary, AB T2N 1N4, Canada}
\email{connor.paddock@ucalgary.ca}

\author{Simon Schmidt}
\affiliation{Faculty of Computer Science, Ruhr University Bochum, 44801 Bochum, Germany}
\email{s.schmidt@rub.de}

\author{Michael Walter}
\affiliation{Faculty of Physics and Faculty of Mathematics, Computer Science and Statistics, Ludwig-Maximilians-Universität München, 80539 München, Germany}
\affiliation{Faculty of Computer Science, Ruhr University Bochum, 44801 Bochum, Germany}
\affiliation{Korteweg-de Vries Institute for Mathematics and QuSoft, University of Amsterdam, 1012 WP Amsterdam, Netherlands}
\email{michael.walter@lmu.de}

\author{Tina Zhang}
\affiliation{Department of Electrical Engineering and Computer Science, Massachusetts Institute of Technology, Cambridge, MA 02139, United States}
\email{tinaz@mit.edu}

\date{}
\maketitle
%-----------------------------------------------------------------------------
\begin{abstract}
Nonlocal games are a foundational tool for understanding entanglement and constructing quantum protocols in settings with multiple spatially separated quantum devices. In this work, we continue the study initiated by Kalai et al. (STOC '23) of \emph{compiled} nonlocal games, played between a classical verifier and a \emph{single} cryptographically limited quantum device. Our main result is that the compiler proposed by Kalai et al.\ is sound for any two-player XOR game. A celebrated theorem of Tsirelson shows that for XOR games, the quantum value is exactly given by a semidefinite program, and we obtain our result by showing that the SDP upper bound holds for the compiled game up to a negligible error arising from the compilation. This answers a question raised by Natarajan and Zhang (FOCS '23), who showed soundness for the specific case of the CHSH game. Using our techniques, we obtain several additional results, including (1) tight bounds on the compiled value of parallel-repeated XOR games, (2) operator self-testing statements for any compiled XOR game, and (3) a ``nice" sum-of-squares certificate for any XOR game, from which operator rigidity is manifest.
\end{abstract}
%=============================================================================

%=============================================================================
\section{Introduction}
%=============================================================================

In a nonlocal game, a polynomially-bounded classical verifier interacts with multiple non-com\-mu\-ni\-ca\-ting provers. A set of predetermined rules governs whether the provers win or lose the game. The \emph{classical value} of a nonlocal game is the highest winning probability that any classical (or unentangled) provers can achieve, whereas the \emph{quantum value} is the highest winning probability that any quantum (potentially entangled) provers can achieve. Although the study of nonlocal games originated with Bell's theorem~\cite{bell1964einstein}, they have become central in the study of quantum complexity theory through their connection with multiprover interactive proofs \cite{CHTW04}. This connection has led to
many protocols for efficiently certifying classical and quantum computations in the two-prover setting~\cite{reichardt2013classical,coladangelo2020verifieronaleash,grilo2017simple,mipre}.

A central question in the nonlocal game paradigm is whether the techniques and results can be translated to a model where a verifier interacts with a \emph{single} prover.
The single-prover setting is theoretically appealing and practically motivated since a non-communication assumption between provers can be challenging to enforce.
Recently, Kalai et al.~\cite{klvy} proposed a generic procedure that compiles \emph{any} $k$-player one-round nonlocal game into a multi-round single-prover protocol, using a quantum homomorphic encryption scheme~\cite{qfhe,brakerskiqfhe}.
In the same work, the authors showed that the compilation process preserves the classical value of the nonlocal game.
Shortly after, Natarajan and Zhang~\cite{nz23} proved that the same compiler preserves the quantum value of the $2$-player CHSH game~\cite{chsh}. This was sufficient to achieve efficient delegation of BQP computations, however, their work left open the question of whether the Kalai et al.~compiler preserved the quantum value for a more general class of games.

A natural generalization of the CHSH game is the class of two-player XOR games~\cite{CHTW04}, where the winning predicate depends on the XOR of the players' single-bit answers. These games have a foundational place in the theory of nonlocal games thanks to an influential theorem of Tsirelson~\cite{Tsirelson,Tsirelson87}, which exactly characterizes their optimal quantum strategies in terms of a semidefinite program. This makes XOR games one of the few classes of nonlocal games where the optimal quantum strategies are well understood, and finding the quantum value is computationally tractable. As such, they are a convenient stepping stone towards understanding more general families of nonlocal games.

\subsection{Results}

Our main result is that the quantum value of any two-player XOR game is preserved by the compilation procedure of~\cite{klvy}, resolving a question posed in \cite{nz23}. We give an informal statement of our result; a more detailed statement can be found in \cref{thm:XORvalue}.
\begin{theorem}\label{thm:main-informal}
    Let $\mathcal{G}$ be any XOR game and $\mathcal{G}_\mathrm{comp}$ the corresponding compiled game. Then any prover running in time polynomial in the security parameter $\lambda$ wins $\mathcal{G}_{\mathrm{comp}}$ with probability at most $\omega_q^*(\mathcal{G}) + \mathrm{negl}(\lambda)$, where $\omega_q^*(\mathcal{G})$ is the quantum value of the original nonlocal game $\mathcal{G}$.
\end{theorem}
The negligible factors in \cref{thm:main-informal} come from the fact that the compiled game assumes the existence of quantum homomorphic encryption, whose security is computational and depends on chosen parameters. It is well-known that quantum homomorphic encryption can be constructed from the hardness of the standard learning with errors (LWE) problem \cite{mahadev,brakerskiqfhe}, and thus our theorem holds up to the same assumption. A useful consequence of this result is that we can apply the extensive literature on XOR games to give new examples of well-behaved compiled games. For instance, we obtain the following new results as almost immediate corollaries of the proof of our main theorem:
\begin{itemize}
    \item \textbf{Self-testing for all XOR games} Using a result from Slofstra \cite{Slofstra11}, we show that every compiled XOR game
    self-tests some properties of the provers' measurement operators: in particular, for each XOR game there is a set of relations between the Alice and Bob operators that are approximately satisfied for any near-optimal strategy for the compiled game (\Cref{thm:XORselftest}). As a special case, we recover the self-test for anticommutation in~\cite{nz23}.
    \\[-2em]
    \item \textbf{Parallel repetition} Building on~\cite{Cleve08}, we show that the quantum value of compiled \emph{parallel repeated} XOR games satisfies perfect exponential decay in the number of parallel repetitions (\Cref{thm:XORparallel}). This is the first parallel repetition theorem for compiled games. %\\[-2em]
\end{itemize}
In addition, we obtain a self-testing result for the compiled Magic Square game, which is not an XOR game but has self-testing properties that are similarly convenient for applications.
\begin{itemize}
    \item \textbf{Self-testing for the Magic Square game:} We show that the compilation of the Magic Square game of Mermin and Peres~\cite{Mermin,Peres,Aravind} is a self-test for a pair of anticommuting operators (\Cref{thm:magic}). This mirrors the statement in the two-prover settings~\cite{WuBancalMcKagueScarani2016}.
\end{itemize}

\subsection{Technical Outline}

\paragraph{Non-local games and the \cite{klvy} compiler.}
We recall the basic setup of a nonlocal game. Before the interaction, Alice and Bob prepare a shared and possibly entangled state $\ket{\psi}$. Once this is done, they are separated and cannot communicate. The verifier samples a pair of challenges (or questions)~$(x,y)$ from a known distribution, and sends $x$ to Alice and $y$ to Bob. The players measure their respective quantum observables~$A_{x}$ and~$B_{y}$ on their respective shares of the state~$\ket{\psi}$, and return outcomes~$a,b$ to the verifier. The verifier decides whether the players win or lose by evaluating a deterministic predicate~$V(x,y,a,b)\in \{0,1\}$.

The KLVY compiler converts any nonlocal game into an interaction between the verifier and a single computationally-bounded prover, the latter of which sequentially emulates the roles of Alice and Bob. The setup of the compiled game is as follows:
\begin{enumerate}
    \item The prover prepares an initial state~$\ket{\psi}$. There is \emph{no} assumption that this state lives in a bipartite Hilbert space. Then the verifier generates a key pair for a quantum fully homomorphic encryption scheme and sends the public key to the prover.
    \item The verifier samples a pair of challenges~$(x,y)$ from a known distribution. It sends the encryption~$c = \Enc(x)$ of the first challenge to the prover. The prover, acting first as Alice, measures the state and returns the outcome~$\alpha$. In the honest case, this measurement is the homomorphic execution of measuring the observable~$A_x$ with respect to $\ket{\psi}$.
    \item The verifier sends~$y$ to the prover in the clear (i.e.~unencrypted). The prover, now acting as Bob, measures an observable~$B_y$ on its residual state and returns the outcome~$b$. There is \emph{no} assumption that this observable commutes with the earlier ``Alice'' measurement.
    \item The verifier decides if the prover wins based on the predicate $V(x,y, \Dec(\alpha), b)$.
\end{enumerate}
It is easy to show that the compiled value is always \emph{at least} the nonlocal value, as any nonlocal strategy can be implemented using the homomorphic encryption scheme. It is far less obvious how to establish upper bounds on the compiled value. This is because all known \emph{general} techniques for upper-bounding the value of a nonlocal game make essential use of the spatial separation between the two provers. Nevertheless, we are able to show that the two specific techniques used by \cite{nz23} to analyze the compiled CHSH game both generalize to all compiled XOR games.
We believe the two approaches are complementary and improve our understanding of compiled nonlocal games from dual perspectives.

% \paragraph{The macroscopic locality approach.}
% Our first approach generalizes the ``macroscopic locality'' proof of soundness of~\cite{nz23} for the CHSH game, which is based on arguments contained in \cite{rohrlich}. One can use the IND-CPA security of the homomorphic encryption to upper-bound the maximum possible probability with which any efficiently measurable Bob operator can guess the value of Alice's question $x$. Intuitively, this allows us to establish that Bob's strategy cannot depend, in any detectable way, on Alice's question. Applying Jensen's inequality to Bob's operator, we can then bound, up to a negligible function, the bias of every quantum strategy for the compiled XOR game by the bias of some \emph{vector} strategy for the \emph{nonlocal} XOR game. These vector strategies can be thought of as a vector solution to the semidefinite program which computes the bias of the XOR game.
% By Tsirelson's theorem, every such vector strategy gives rise to a quantum strategy for the nonlocal XOR game, yielding the desired bound on the winning probability for the compiled game in terms of the optimal nonlocal value of the original XOR game.
% As a corollary, we obtain a one-sided rigidity or ``self-testing" statement for any compiled XOR game: For any near-optimal strategy, it must hold that $A_{x}\ket{\psi} \approx \sum_{y} c_{xy} B_{y}A_{x} \ket{\psi}$ for explicit coefficients $c_{xy}$.

%\paragraph{The sum of squares approach.}
An approach to upper-bounding the quantum value of a nonlocal game is to associate with a game a formal polynomial~$p$, whose variables are identified with the family of  observables~$\{A_{x}\}$ and~$\{B_{y}\}$ so that $\bra{\psi} p\left(\{A_x\},\{B_y\}\right) \ket{\psi}$ is precisely the winning probability of the strategy. (For example, in the case of CHSH, the polynomial is $p\left(\{A_x\},\{B_y\}\right) = \frac{1}{2} + \frac{1}{8}(A_0 B_0 + A_0 B_1 + A_1 B_0 - A_1 B_1)$, and it can be easily verified that $\bra{\psi} p\left(\{A_x\},\{B_y\}\right) \ket{\psi}$ captures the winning probability of the strategy $(\{A_x\}, \{B_y\}, \ket{\psi})$ for the CHSH game. Moreover, to prove that the quantum value of a game is at most $\omega$, it suffices to give a sum-of-squares (SOS) certificate demonstrating that $\omega\Id - p\left(\{A_x\},\{B_y\}\right)$ is a positive operator whenever $p$ is evaluated on variables that corresponds to valid observables. Such a certificate consists of a decomposition of $\omega\Id - p\left(\{A_x\},\{B_y\}\right)$ as a positive operator
\[ \omega\Id - p\left(\{A_x\},\{B_y\}\right) = \sum_i q_i\left(\{A_x\},\{B_y\}\right)^* q_i\left(\{A_x\},\{B_y\}\right) + g\left(\{A_x\},\{B_y\}\right),\]
where the $q_i$'s and $g$ are polynomials, and $g$ is a polynomial that evaluates to zero whenever the variables are valid observables. In particular, $g$ lies in the ideal generated by a certain collection of polynomial constraints on the variables $A_{x}$ and $B_{y}$, including the commutator constraint $[A_{x},B_{y}] = 0$ for all~$x,y$, which formally represents the spatial separation of the players.

The reason why such an SOS certificate suffices to upper bound the game value in the nonlocal case is that (1) taking the expectation $\bra{\psi} \omega\Id - p\left(\{A_x\},\{B_y\}\right) \ket{\psi}$ maps positive operators to positive numbers, and (2) $\bra{\psi} p\left(\{A_x\},\{B_y\}\right) \ket{\psi}$ is equal to the winning probability of strategy $S=(\{A_x\}, \{B_y\}, \ket{\psi})$. Property~(1) gives that evaluating the expectation on $\omega\Id - p\left(\{A_x\},\{B_y\}\right)$ yields a positive number, and property~(2) says that this positive number is equal to $\omega - p_{win}$, where $p_{win}$ is the winning probability of the strategy $S=(\{A_x, B_y\}, \ket{\psi})$. Therefore, in the nonlocal case, taking the expectation plays the role of what is often known as a ``pseudo-expectation'' in the SOS literature.

In general, an SOS certificate of the aforementioned form does not automatically yield an upper bound on the compiled value. The reason is that there is no longer an obvious choice of pseudo-expectation operator which will satisfy the requisite properties. In particular, the game value can no longer be described by a sum of quantities that look like $\bra{\psi}A_x B_y\ket{\psi}$, as we no longer have the guarantee $[A_x, B_y] = 0$ afforded to us before by the spatial separation, meaning that the~$g$ polynomial may not be zero when evaluated on the ``valid'' strategies for the compiled game. In~\cite{nz23}, a choice of pseudo-expectation operator was carefully defined so that the pseudo-expectation of~$g\left(\{A_x\},\{B_y\}\right)$ was zero by fiat, and (weaker versions of) properties~(1) and~(2) could be argued by reduction to cryptographic security properties. However, the fact that the measurement operators used by Alice in the compiled game return encrypted outcomes was an obstacle to carrying out these cryptographic reductions because the security only holds with respect to efficient distinguishers.
In \cite{nz23}, this obstacle was overcome in the CHSH game because the SOS certificate for this game has a particularly simple dependence on the $A_{x}$ observables.
Building on a work of Slofstra~\cite{Slofstra11}, we show that \emph{any} XOR game has a ``nice" SOS certificate (in which each $q_i$ has degree $1$ and depends on at most one $A_{x}$ variable) and that these ``nice'' SOS certificates yield bounds on the compiled value as well as the quantum value. As a corollary, we obtain a one-sided rigidity or ``self-testing" statement for any compiled XOR game: For any near-optimal strategy, it must hold that $A_{x}\ket{\psi} \approx \sum_{y} c_{xy} B_{y}A_{x} \ket{\psi}$ for explicit coefficients $c_{xy}$. 

Although the approach using SOS allows us to obtain a bound on the compiled value for any XOR game, we mention that there is an alternative proof of this fact that employs Tsirelson's correspondence in the compiled setting. Tsirelson's correspondence shows that every vector strategy gives rise to a quantum strategy for the nonlocal XOR game \cite{Tsirelson87,Cleve08}. Here, by applying Jensen's inequality to Bob's operators, we can bound (up to a negligible function) the bias of every quantum strategy for the compiled XOR game by the bias achieved by some \emph{vector strategy} for the nonlocal XOR game. Therefore, by Tsirelson's correspondence we obtain a bound on the winning probability for the compiled game in terms of the optimal nonlocal value of the original XOR game. This approach was also used in ~\cite{nz23} for the CHSH game and inspired the generalization here. However, as this approach does not recover the proof of the ``nice'' SOS certificates for XOR games we defer the alternate proof to the \cref{appendix}.

\subsection{Further directions} The most immediate question raised by our work is to determine how far our techniques can extend beyond XOR games: for instance, can we find ``nice" SOS decompositions, of higher degree, for all nonlocal games? Another more conceptual question is to determine the ``correct" notion of rigidity in the compiled setting: unlike in the nonlocal case, our results are not symmetric between between Alice and Bob.
 \subsection{Related work}
 During the completion of this work, we became aware that the self-testing result for the Magic Square game in \cref{sec:magic} was independently obtained in unpublished works by Thomas Vidick, and separately Fermi Ma and Alex Lombardi, using similar techniques to ours. Shortly after posting the first draft of this work, the independent work \cite{baroni2025quantum} appeared and obtained the same result for the value of compiled XOR games.

%=============================================================================
\section{Preliminaries}
%=============================================================================
%-----------------------------------------------------------------------------
\subsection{Cryptography}
%-----------------------------------------------------------------------------
We adopt several definitions from~\cite{nz23}.

\begin{definition}
A \emph{QPT (quantum polynomial time) algorithm} is a logspace-uniform family of quantum circuits with size polynomial in the number of input qubits and in the security parameter.
If the circuits are unitary then we call this a \emph{(unitary) QPT circuit}.
A POVM~$\{M_\beta\}$ is called \emph{QPT-measurable} is there is a QPT circuit such that measuring some output qubits and post-processing gives rise to the same probabilities as the POVM.
% For a binary POVM,  we can always assume that we measure a single output qubit and do no post-processing.
A binary observable~$B$ is called \emph{QPT-measurable} if this is the case for the corresponding projective POVM.
This is equivalent (by uncomputation) to demanding that~$B$ interpreted as a unitary can be realized by a QPT circuit.
\end{definition}
Here we follow~\cite{nz23} in considering security against uniform adversaries, and indeed all of our reductions are uniform.
We remark that we can also define security against non-uniform adversaries to obtain a stronger conclusion at the cost of relying on a stronger cryptographic assumption (specifically, QHE secure against non-uniform adversaries, which is quite standard in cryptography).
We now recall the notion of quantum homomorphic encryption (QHE).

\begin{definition}
    A \emph{quantum homomorphic encryption scheme} $\mathrm{QHE}=(\mathrm{Gen}, \mathrm{Enc}, \mathrm{Eval}, \mathrm{Dec})$ for a class of circuits $\mathcal{C}$ is defined as a tuple of algorithms with the following syntax.
    \begin{itemize}
        \item $\mathrm{Gen}$ is a PPT algorithm that takes as input the security parameter $1^\lambda$ and returns a secret key $\mathrm{sk}$.
        \item $\mathrm{Enc}$ is a PPT algorithm that takes as input a secret key $\mathrm{sk}$ and a classical input~$x$, and outputs a classical ciphertext~$c$.
        \item  $\mathrm{Eval}$ is a QPT algorithm that takes as input a tuple a classical description of a quantum circuit~$C: \mathcal{H} \times (\mathbb{C}^2)^{\otimes n} \to (\mathbb{C}^2)^{\otimes m}$, a quantum register with Hilbert space~$\mathcal{H}$, and a ciphertext~$c$, % given by an $n$-bit plaintext.
        and outputs a ciphertext~$\Tilde{c}$.
        If $C$ has classical output, we require that~$\mathrm{Eval}$ also has classical output.
        \item $\mathrm{Dec}$ is a QPT algorithm that takes as input a secret key $\mathrm{sk}$ and ciphertext $c$, and outputs a state $\ket{\psi}$.
        Additionally, if $c$ is a classical ciphertext, the decryption algorithm deterministically outputs a classical string $y$.
    \end{itemize}
    We require that the scheme satisfies the following properties.
    \begin{itemize}
        \item \emph{Correctness with Auxiliary Input:}
        For any $\lambda\in\mathbb{N}$, any quantum circuit $C: \mathcal{H}_A \times (\mathbb{C}^2)^{\otimes n} \to \{0,1\}^*$ with classical output, any state $\ket{\psi}_{AB} \in \mathcal{H}_A\otimes \mathcal{H}_B$, and any message $x\in\{0,1\}^n$, the following experiments output states with negligible trace distance:
%        \MW{Please check.}
        \begin{itemize}
            \item Game 1:
            Start with~$(x,\ket{\psi}_{AB})$.
            Apply the circuit~$C$, obtaining the the classical string~$y$.
            Return~$y$ and register~$B$.
            \item Game 2:
            Start with a key $\mathrm{sk}\leftarrow\mathrm{Gen}(1^\lambda)$, $c\in\mathrm{Enc}(\mathrm{sk}, x)$, and~$\ket{\psi}_{AB}$.
            Apply $\mathrm{Eval}$ with input~$C$, register~$A$, and ciphertext~$c$ to obtain~$\Tilde{c}$.
            Compute $\Tilde y \leftarrow \mathrm{Dec}(\mathrm{sk}, \Tilde{c})$.
            Return~$\Tilde y$ and register $B$.
        \end{itemize}
        \item \emph{CPA Security:} For all pairs of messages $(x_0, x_1)$ and any QPT adversary $\mathcal{A}$ it holds that
        \[
        \left| \Pr\left[1=\mathcal{A}(c_0)^{\mathrm{Enc}(\mathrm{sk}, \cdot)}\middle|
        \begin{array}{l}
             \mathrm{sk} \gets  \mathrm{Gen}(1^\lambda) \\
             c_0 \gets \mathrm{Enc}(\mathrm{sk}, x_0)
        \end{array}
        \right]-
        \Pr\left[1=\mathcal{A}(c_1)^{\mathrm{Enc}(\mathrm{sk}, \cdot)}\middle|
        \begin{array}{l}
             \mathrm{sk} \gets  \mathrm{Gen}(1^\lambda) \\
             c_1 \gets \mathrm{Enc}(\mathrm{sk}, x_1)
        \end{array}
        \right]
        \right| \leq \mathrm{negl}(\lambda),
        \]
    \end{itemize}
\end{definition}
\begin{remark}\label{remark:abuse of notation}
    In a slight abuse of notation, we often write expressions such as~$\mathop{\mathbb{E}}_{c\leftarrow\mathrm{Enc}(x)} f(\mathrm{Dec}(\alpha))$ as an abbreviation for an expectation value of the form~$\mathop{\mathbb{E}}_{\mathrm{sk} \gets \mathrm{Gen}(1^\lambda), c\leftarrow\mathrm{Enc}(\mathrm{sk}, x)} f(\mathrm{Dec}(\mathrm{sk}, \alpha))$.
\end{remark}

 % sometimes write $c = \mathsf{Enc}(m)$ to denote the fact that $c$ is in the support of the total domain of the ciphertexts encrypting some (classical) message $m$.

%-----------------------------------------------------------------------------
\subsection{Non-Local Games}
%-----------------------------------------------------------------------------
% \MW{If we do this, should do it before defining QPT. Perhaps not necessary for this note.}
% \begin{itemize}
    % \item \color{red}{Definition state, PVM, observables}
% \end{itemize}

\begin{definition}
A \emph{nonlocal game} $\mathcal{G}$ is a tuple $(I_A, I_B, O_A, O_B, \pi, V)$ consisting of finite sets~$I_A$ and~$I_B$ of inputs for Alice and Bob, respectively, finite sets~$O_A$ and~$O_B$ of outputs for Alice and Bob, respectively, a probability distribution of the inputs~$\pi: I_A\times I_B\to [0,1]$, and a verification function $V: I_A\times I_B\times O_A\times O_B \to \{0,1\}$.
\end{definition}

A nonlocal game is played by a verifier and two provers, Alice and Bob.
In the game, the verifier samples a pair $(x,y) \leftarrow \pi$ and sends $x$ to Alice and $y$ to Bob.
Alice and Bob respond with $a\in O_A$ and $b\in O_B$, respectively.
They win if $V (x, y, a, b) = 1$.
The players are not allowed to communicate during the game, but they can agree on a strategy beforehand.
Their goal is to maximize their winning probability.
If we do not specify otherwise, the distribution $\pi$ will be the uniform distribution.

\begin{definition} \label{def:tensor_strat}
A \emph{quantum (tensor) strategy} $S$ for a nonlocal game $\mathcal{G}=(I_A, I_B, O_A, O_B, \pi, V)$ is a tuple $S=(\mcH_A,\mcH_B,\ket{\psi}, \{A_{xa}\}, \{B_{yb}\})$, consisting of finite dimensionals Hilbert spaces~$\mcH_A$ and~$\mcH_B$, a bipartite state $\ket{\psi}\in \mathcal{H}_A\otimes \mathcal{H}_B$, PVMs $\{A_{xa}\}_{a\in O_A}$ acting on $\mathcal{H}_A$ for each $x \in I_A$ for Alice and PVMs $\{B_{yb}\}_{b\in O_B}$ acting on $\mathcal{H}_B$ for each $y \in I_B$ for Bob. Often we will drop the Hilbert spaces, and just write $S=(\ket{\psi}, \{A_{xa}\}, \{B_{yb}\})$.
\end{definition}

Here we restrict without loss of generality to pure states and projective measurements (PVMs).
For a strategy $S=(\ket{\psi}, \{A_{xa}\}, \{B_{yb}\})$, the probability of Alice and Bob answering $a, b$ when obtaining $x,y$ is given by $p(a,b|x,y)=\bra{\psi} A_{xa}\otimes B_{yb}\ket{\psi}$. Therefore, the \emph{winning probability} of a quantum strategy~$S$ for the nonlocal game $\mathcal{G}$ is given by
\begin{align*}
    \omega_q(S,\mathcal{G})=\sum_{x,y} \pi(x,y) \sum_{a,b}V(x,y,a,b)p(a,b|x,y)=\sum_{x,y} \pi(x,y) \sum_{a,b}V(x,y,a,b)\bra{\psi} A_{xa}\otimes B_{yb}\ket{\psi}.
\end{align*}
For a nonlocal game $\mathcal{G}$, we define the \emph{quantum value} $\omega_q^*(\mathcal{G})=\mathrm{sup}_{S}\omega_q(S,\mathcal{G})$ to be the supremum over all quantum tensor strategies compatible with $\mathcal{G}$. A strategy $S$ is called \emph{$\varepsilon$-optimal} for a game $\mathcal{G}$, if it holds $\omega_q(S,\mathcal{G})\geq \omega_q^*(\mathcal{G})-\varepsilon$. In the case of $\varepsilon=0$, i.e. $\omega_q(S,\mathcal{G})=\omega_q^*(\mathcal{G})$, we usually omit $\varepsilon$ and call the strategy \emph{optimal}.

The tensor-product structure is a way of mathematically representing the locality of the players employing a quantum strategy in a nonlocal game. However, there is a more general way to model this nonlocality mathematically.

\begin{definition}\label{commuting_strat}
    A \emph{commuting operator strategy} $\mathcal{S}$ for a nonlocal game $\mathcal{G}=(I_A, I_B, O_A, O_B, \pi, V)$ is a tuple $\mathcal{S}=(\mcH,\ket{\psi}, \{A_{xa}\}, \{B_{yb}\})$, consisting of a Hilbert space $\mcH$, a state $\ket{\psi}\in \mathcal{H}$, and two collections of mutually commuting PVMs $\{A_{xa}\}_{a\in O_A}$ acting on $\mathcal{H}$ for each $x \in I_A$ for Alice and PVMs $\{B_{yb}\}_{b\in O_B}$ acting on $\mathcal{H}$ for each $y \in I_B$ for Bob, i.e. $[A_{xa},B_{yb}]=0$ for all $a,b,x,y\in O_A\times O_B\times I_A\times I_B$. Like for quantum strategies, we will often omit the Hilbert space and write $\mathcal{S}=(\ket{\psi}, \{A_{xa}\}, \{B_{yb}\})$ for a commuting operator strategy.
\end{definition}

We can also define the commuting operator (also known as the quantum commuting) value of a nonlocal game $\omega_{qc}^*(\mathcal{G})=\sup_{S}\omega_{qc}(\mathcal{S},\mathcal{G})$ to be the supremum over all commuting operator strategies~$S$ compatible with $\mathcal{G}$. It is not hard to see that every quantum (tensor) strategy is a commuting operator strategy. The converse holds if we restrict our commuting operator strategies to be finite dimensional (i.e. $\mcH$ is finite dimensional). However, there are examples of nonlocal games $\mcG$ for which there is a perfect (wins with probability $1$) commuting operators strategy but no perfect tensor-product strategy, see for example \cite{Slofstra16}.

%-----------------------------------------------------------------------------
\subsection{XOR Games}
%-----------------------------------------------------------------------------

\begin{definition}
An \emph{XOR game} is a nonlocal game for which $O_A=O_B=\{0,1\}$ and $V(x,y,a,b)=\frac{1}{2}(1+(-1)^{g(x,y)}(-1)^{a + b})$ for some function $g: I_A\times I_B\to \{0,1\}$. An XOR game is completely described by the \emph{cost matrix} $G=(G_{xy})_{x\in I_A, y \in I_B}$, where $G_{xy}=(-1)^{g(x,y)}\pi(x,y).$ We will say that~$\mathcal{G}_\mathrm{comp}$ is a \emph{compiled XOR game} if it is a compiled game for an XOR game $\mathcal{G}$.
\end{definition}

%For a strategy $S$ of an XOR game $\mathcal{G}$, it is sometimes advantageous to work with the quantum \emph{bias}
%\[ \beta(S,\mathcal{G})=2\omega_q(S,\mathcal{G})-1 \]
%instead of the winning probability $\omega(S,\mathcal{G})$, since the bias naturally comes up if we express the PVMs in $S$ in terms of binary observables.

Given a commuting operator strategy for an XOR game $\mathcal{S}=(\ket{\psi}, \{A_{xa}\}, \{B_{yb}\})$ we can derive an expression for the commuting operator \emph{bias} of the strategy $\beta_{qc}(\mathcal{S},\mcG)$ in terms of $\pm1$-valued observables (i.e. self-adjoint unitaries) as follows:
\begin{align}
2\omega_{qc}(\mathcal{S},\mcG)-1&=2\left[\sum_{x,y}\pi(x,y)\sum_{a,b}V(a,b,x,y)\langle\psi|A_{xa} B_{yb}|\psi\rangle \right]-1\\
&=2\left[\sum_{x,y}\pi(x,y)\sum_{a,b}\frac{1}{2}(1+(-1)^{g(x,y) + a + b})\langle\psi|A_{xa} B_{yb}|\psi\rangle \right]-1\\
&=\sum_{x,y}\pi(x,y)\sum_{a,b}(-1)^{g(x,y) + a + b}\langle\psi|A_{xa} B_{yb}|\psi\rangle\\
&=\sum_{x,y}\pi(x,y)(-1)^{g(x,y)}\langle\psi|\left(\sum_a (-1)^a A_{xa}\right)\left( \sum_b (-1)^b B_{yb}\right)|\psi\rangle\\
&=\sum_{x,y}G_{xy}\langle\psi|A_{x} B_{y}|\psi\rangle\\
&=\beta_{qc}(\mathcal{S},\mcG),
\end{align}
where $A_x=A_{x0}-A_{x1}$ and $B_{y}=B_{y0}-B_{y1}$ are the $\pm1$-valued (binary) observables for all $x,y\in I_A\times I_B$. Note that this construction is reversible, and thus strategies consisting of these $\pm1$-valued observables are equivalent to commuting operator strategies for $\mcG$.
With this in mind, we define the \emph{optimal commuting operator bias} of $\mathcal{G}$ to be $\beta_{qc}^*(\mathcal{G})=2\omega_{qc}^*(\mathcal{G})-1$. Since all quantum strategies are commuting operator strategies we can similarly define the quantum \emph{bias} $\beta_q(S,\mathcal{G})=2\omega_q(S,\mathcal{G})-1$. We define the \emph{optimal quantum bias} of $\mathcal{G}$ to be $\beta_q^*(\mathcal{G})=2\omega_q^*(\mathcal{G})-1$.

%Recall an XOR game $\mcG$ is specified by a cost matrix $G \in \R^{n \times m}$ whose entries are $\pm 1$ or $0$. That is rows are indexed by Alice's question alphabet $I_A=[n]$ and the columns by Bob's question alphabet $I_B=[m]$. The entries are given by $G_{xy}=(-1)^{f(x,y)}\pi(x,y)$, where $\pi(x,y)$ is the probability distribution on questions and $f:X\times Y \to \{0,1\}$ is determined by the XOR game predicate. Remark that we have $G_{xy} = 0$ if the question pair $(x,y)$ does not occur. We will always assume that the question distribution is uniform over the set of allowed question pairs.

We shall now present some of the relevant structural properties we will need about XOR games.

\begin{definition}\label{def:XOR_correlation}
    For an XOR game $\mcG$, a \emph{vector strategy} $\mathcal{V}$ is a Euclidean space $\R^d$, and a collection of real unit vectors $\cbrac{u_{x}:x\in I_A}\in \R^d$ and $\cbrac{v_{y}:y\in I_B}\in \R^{d}$. The \emph{vector bias} for $\mcG$ is given by
    \begin{equation}
        \beta(\mathcal{V},\mcG)=\sum_{x,y}G_{xy}\langle u_x|v_y\rangle.
    \end{equation}
    The optimal vector bias is denoted by $\beta^*_{vec}(\mcG)=\sup_{\mathcal{V}}\beta(\mathcal{V},\mcG)$, and $\mathcal{V}$ is an optimal vector strategy if $\beta_{vec}(\mathcal{V},\mcG)=\beta^*_{vec}(\mcG)$.
    \end{definition}

\begin{definition}
Let $ S= \{ \ket{\psi}, \{A_{x}\},\{B_{y}\}\}$ be a quantum strategy for an XOR game $\mcG$. We define the \emph{row biases} $\{r_{x,S}\in \R:x\in I_A\}$, where each
    \[ r_{x,S} = \sum_y G_{xy} \bra{\psi} A_x B_y \ket{\psi}.\] We similarly define the \emph{column biases} $\{c_{y,S}\in \R:y\in I_B\}$, where each
    \[ c_{y,S} = \sum_x G_{xy} \bra{\psi} A_x B_y \ket{\psi}.\]
\end{definition}

These bias quantities play an important role in the following theorem:

\begin{theorem}[\cite{Slofstra11,Tsirelson87}]\label{thm:Tsirelson} \label{thm:slofstra}
For any XOR game $\mcG$, the following hold:
\begin{enumerate}
    \item All optimal quantum strategies $S$ have the same bias $\beta_q(S)$, row biases $\{r_{x,S}\}$ and column biases $\{c_{y,S}\}$, henceforth denoted $\beta$, $\{r_x:x\in I_A\}$, and $\{c_y:y\in I_B\}$ respectively. All optimal row and column biases are positive.
    \item In any optimal quantum commuting operator strategy, we have that
\[ A_x \ket{\psi} = \hat{B}_x \ket{\psi} := \frac{1}{r_x} \sum_y G_{xy} B_y \ket{\psi}, \] for all $x\in I_A$.
\item The optimal quantum bias $\beta^{\ast}_{q}(\mathcal{G})$ is the solution to the following primal and dual semi-definite programs:
\begin{equation*}
    \begin{aligned}[t]
        (P): \quad \mathop{\mathrm{max}}_{\substack{u_x,v_y\in \R^{d}\\\|u_x\|=\|v_y\|=1}} \sum_{x \in I_A,y \in I_B} G_{xy}\braket{u_x}{v_y}
    \end{aligned}
    \qquad\quad
    \begin{aligned}[t]
        (D): \quad \min_{\{r_x\}, \{c_y\}} &\frac{1}{2}\sum_{x} r_x + \frac{1}{2}\sum_y c_y :  \\  &\begin{pmatrix} \Delta(\mathbf{r}) & -G^{T} \\ -G & \Delta(\mathbf{c}) \end{pmatrix} \succeq 0,
    \end{aligned}
\end{equation*}
where $\Delta(\mathbf{r})$ is the $|I_A|\times |I_A|$ diagonal matrix with the values $r_x$ along the diagonal, and likewise $\Delta(\mathbf{c})$ is the $|I_B| \times |I_B|$ diagonal matrix with the values $c_y$ along the diagonal. Furthermore, the optimal row and column biases are the solutions to the dual semi-definite program.

\end{enumerate}
\end{theorem}
\begin{proof}
    The first two items are Theorem~3.1 and Corollary~3.2 of \cite{Slofstra11}, together with the fact that any optimal quantum strategy is an optimal vector strategy, which follows from \cref{xorStrategies}. The formulation of the primal in (3) is a direct consequence of \cref{xorStrategies}, and the duality of the semi-definite programs is shown in the proof of Theorem 3.1 of~\cite{Slofstra11}.
\end{proof}

\begin{corollary}\label{xorStrategies}
    Let $\beta \in [-1,1]$. If $\mcG$ is an XOR game, then the following are equivalent:
    \begin{enumerate}
    \item $\mcG$ has an optimal quantum strategy $S$ with $\beta_q(S, \mcG)=\beta$.
    \item $\mcG$ has an optimal commuting operator strategy $\mathcal{S}$ with $\beta_{qc}(\mathcal{S}, \mcG)=\beta$.
    \item $\mcG$ has an optimal vector strategy $\mathcal{V}$ with $\beta_{vec}(\mathcal{V}, \mcG)=\beta$.
    \end{enumerate}
\end{corollary}

%%%%%%%%%%%%%%%%%%%%%%%%%%%%%%%%

%-----------------------------------------------------------------------------
\subsection{Compiled Games}
%-----------------------------------------------------------------------------
\begin{definition}\label{def:compiledgame}
A \emph{compiled game} $\mathcal{G}_\mathrm{comp}$ consists of a nonlocal game $\mathcal{G}$ and a quantum homomorphic encryption scheme $\mathrm{QHE}=(\mathrm{Gen}, \mathrm{Enc}, \mathrm{Eval}, \mathrm{Dec})$. However, unlike a standard nonlocal game, it is played by a verifier and a single prover. The behaviour of the interaction is described as follows:
\begin{itemize}
    \item[1.] The verifier samples $(x,y)\leftarrow \pi$, $\mathrm{sk}\leftarrow \mathrm{Gen}(1^\lambda)$, and $c\leftarrow \mathrm{Enc}(\mathrm{sk},x)$. The verifier sends~$c$ to the prover.
    \item[2.] The prover replies with some classical ciphertext~$\alpha$.
    \item[3.] The verifier sends~$y$ (in the clear) to the prover.
    \item[4.] The prover replies with some classical message~$b$.
    \item[5.] The verifier computes $a:=\mathrm{Dec}(\mathrm{sk},\alpha)$ and accepts if and only if~$a\in O_A$, $b\in O_B$, and $V(a,b,x,y)=1$.
\end{itemize}
\end{definition}

\begin{definition}
A \emph{quantum strategy for a compiled game} is a tuple $(\mcH,\ket{\psi}, \{A_{c\alpha}\}, \{B_{yb}\})$ consisting of a Hilbert space $\mcH$, an efficiently
(in QPT) preparable state $\ket{\psi}\in \mathcal{H}$, operators~$A_{c\alpha}=U_{c\alpha}P_{c\alpha}$ acting on~$\mathcal{H}$, where~$U_{c\alpha}$ is a QPT-measurable unitary and~$\{P_{c\alpha}\}_{\alpha\in \Lambda}$ (where $\Lambda$ is the set of outcomes), is a QPT-measurable PVM for all~$c\in \Enc(sk,x)$, $x\in I_A$, as well as QPT-measurable PVMs~$\{B_{yb}\}_{b\in O_B}$ acting on~$\mathcal{H}$ for all~$y\in I_B$.
\end{definition}

For convenience, we let $\ket{\psi_{c\alpha}}:=A_{c\alpha}\ket{\psi}$ be the \emph{unnormalized post-measurement state} after Step~2 in \cref{def:compiledgame}. For a strategy for the compiled game $S=(\ket{\psi}, \{A_{c\alpha}\}, \{B_{yb}\})$, the probability of Alice and Bob answering $a, b$ after being given $x,y$ is denoted
\begin{align*}
    p(a,b|x,y)=\mathop{\mathbb{E}}_{c\leftarrow\mathrm{Enc}(x)}\sum_{\alpha; \mathrm{Dec}(\alpha)=a}\bra{\psi} (A_{c\alpha})^*B_{yb}A_{c\alpha}\ket{\psi}
    =\mathop{\mathbb{E}}_{c\leftarrow\mathrm{Enc}(x)}\sum_{\alpha; \mathrm{Dec}(\alpha)=a}\bra{\psi_{c\alpha}}B_{yb}\ket{\psi_{c\alpha}}
\end{align*}
It follows that the \emph{winning probability} of the quantum strategy~$S$ for the compiled game~$\mathcal{G}_\mathrm{comp}$ is given by
\begin{align}\label{eq:compiled winning prob}
    \omega_q(S,\mathcal{G}_\mathrm{comp})=\sum_{x,y} \pi(x,y) \sum_{a,b}V(x,y,a,b)\mathop{\mathbb{E}}_{c\leftarrow\mathrm{Enc}(x)}\sum_{\alpha; \mathrm{Dec}(\alpha)=a}\bra{\psi_{c\alpha}}B_{yb}\ket{\psi_{c\alpha}}.
\end{align}

%Similar to the quantum value of a nonlocal game, we therefore define the \emph{quantum value}~$\omega_q^*(\mathcal{G}_\mathrm{comp})=\sup_S\omega_q(S,\mathcal{G}_\mathrm{comp})$ of the compiled game to be the supremum of the winning probability over all compatible quantum strategies for $G_\mathrm{comp}$.
%Note that both are a function of the security parameter~$\lambda$.
%\MW{Actually, I don't think we can give an upper bound on the quantum value defined as a $\lambda$-pointwise supremum.}

\begin{theorem}\label{thm:KLVY}(\cite[Theorem 3.2]{klvy})
If $\mathcal{G}_\mathrm{comp}$ is a compiled nonlocal game with underlying nonlocal game~$\mathcal{G}$. Then, there exists a compiled quantum strategy $S$ for $\mathcal{G}_\mathrm{comp}$ and a negligible function $\eta(\lambda)$ such that
\begin{align*}
    \omega_q(S, \mathcal{G}_\mathrm{comp})\geq \omega_q(\mathcal{G})-\eta(\lambda),
\end{align*}
where $\eta(\lambda)$ depends on $S$.
\end{theorem}

%-----------------------------------------------------------------------------
\subsection{Block encodings}
%-----------------------------------------------------------------------------
Given binary observables that are efficiently measurable, we will also be interested in measuring their linear combinations and products in order to infer relations.
To this end, we use block encodings as in~\cite{nz23}.
See~\cite{gilyen2019quantum} for a general overview of this notion.
We only define it for square matrices:

\begin{definition}
A \emph{block encoding} of a matrix~$H$ on~$(\C^2)^{\ot n}$ is a unitary~$U$ on~$(\C^2)^{\ot (m+n)}$ for some~$m \in \N$ such that
\begin{align*}
    t H = \parens*{ \bra0^{\ot m} \ot I } U \parens*{ \ket0^{\ot m} \ot I },
    \qquad\text{i.e.}\qquad
    U = \begin{pmatrix}
      tH & * \\
      * & *
    \end{pmatrix}
\end{align*}
for some~$t>0$ called the \emph{scale factor} of the block encoding.
(Note that any unitary is trivially block encodable with~$m=0$ and~$t=1$.)
A block encoding is called \emph{QPT-implementable} if~$U$ is implemented by a QPT circuit.
\end{definition}

The below lemmas summarize and generalize the discussion in~\cite{nz23} (cf.~Lemmas~11--13 therein).

\begin{lemma}\label{lem:block from binary}
If~$B = B_+ - B_-$ is a QPT-measurable binary observable, with projections~$B_\pm$, then%
% \begin{enumerate}% \item
~$B$ interpreted as a unitary can be implemented by a QPT circuit (in particular, it is a QPT-implementable block encoding of itself, with scale factor~1).
% \item $B_\pm$ have QPT-implementable block encodings with scale factor~$\Theta(1)$.
% \end{enumerate}
\end{lemma}

%\MW{Note: In~\cite{nz23} they don't actually state or use the second item. I added it for convenience in case we need it. Shall we comment it out?}

\begin{lemma}\label{lem:block lincom}
For any fixed~$x \in \R^k$, the following holds.
If $H_1,\dots,H_k$ are operators with QPT-implementable block encodings, each with scale factor~$\Theta(1)$, then~$H = \sum_{j=1}^k x_j H_j$ also has a QPT-implementable block encoding with scale factor~$\Theta(1)$.
\end{lemma}

\Cref{lem:block lincom} follows from~\cite[Lemma~29]{gilyen2019quantum}.
We note that it can be extended to~$k=\operatorname{poly}(\lambda)$ many block encodings (if the vector~$x$ can be prepared efficiently).

\begin{lemma}\label{lem:block prod}
If~$H$ and~$H'$ are operators with QPT-implementable block encodings, each with scale factor~$\Theta(1)$, then their product~$H H'$ also has a QPT-implementable block encoding with scale factor~$\Theta(1)$.
\end{lemma}

% See~\cite[Lemma~12]{nz23} and~\cite[Lemma~30]{gilyen2019quantum}.
In the remainder of this section we will discuss how to get a strong kind of non-signalling property for strategies of compiled games.
We start by recalling two lemmas from~\cite{nz23}:

\begin{lemma}\label{lemma:StrongnonsignallingPOVM}(Lemma 9 in~\cite{nz23})
Let $\lambda \in \N$ be a security parameter. For any two efficiently (in QPT) sampleable distributions $D_1$, $D_2$ over plaintext Alice questions, for any efficiently preparable state $\ket{\psi}$, and for any QPT-measurable POVM $\{M_{\beta}\}_{\beta}$ with outcomes in $[0, 1]$, there exists a negligible function $\eta'(\lambda)$ such that
\begin{align*}
    \left|\mathop{\mathbb{E}}_{x\leftarrow D_1}\mathop{\mathbb{E}}_{c\leftarrow\mathrm{Enc}(x)}\sum_{\alpha}\sum_{\beta}\beta\bra{\psi}A_{c\alpha}^* M_{\beta}A_{c\alpha}\ket{\psi}-\mathop{\mathbb{E}}_{x\leftarrow D_2}\mathop{\mathbb{E}}_{c\leftarrow\mathrm{Enc}(x)}\sum_{\alpha}\sum_{\beta}\beta\bra{\psi} A_{c\alpha}^* M_{\beta}A_{c\alpha}\ket{\psi}\right|\leq \eta'(\lambda).
\end{align*}
\end{lemma}

\begin{lemma}\label{lem:QPTmeasureableuptoepsilon}(Lemma 14 in~\cite{nz23})
Suppose we have a QPT-implementable block encoding with scale factor~$\Theta(1)$ for a (not necessarily binary) observable $\mathcal{B}$, and suppose that $\|\mathcal{B}\|\leq \Theta(1)$. Then $\mathcal{B}$ is QPT-measurable up to precision $\varepsilon$ for every $\varepsilon^{-1}=\operatorname{poly}(\lambda)$, that is, there exists a QPT-measurable POVM $\{M_{\beta}\}_{\beta}$ such that for any state $\rho$ it holds
\begin{align*}
    \left|\sum_{\beta}\beta \operatorname{tr}(M_{\beta}\rho)-\operatorname{tr}(\mathcal{B}\rho)\right|\leq \varepsilon.
\end{align*}
%\MW{Restate~\cite[Lemma~14]{nz23}, since we need this to eventually get a contradiction to \cref{lemma:Strongnonsignalling} (=~\cite[Lemma~9]{nz23}).}
\end{lemma}

The next lemma is stated more generally than~\cite[Lemma 15--17]{nz23}, but the proof is essentially the same.

\begin{lemma}\label{lemma:Strongnonsignalling}
Let $\lambda \in \N$ be a security parameter and suppose we have a QPT-implementable block encoding with scale factor~$\Theta(1)$ for an observable $\mathcal{B}$ with $\|\mathcal{B}\|= O(1)$.
For any two efficiently (in QPT) sampleable distributions $D_1$, $D_2$ over plaintext Alice questions and for any efficiently preparable state $\ket{\psi}$, there exists a negligible function $\eta'(\lambda)$ such that it holds
\begin{align*}
    \left|\mathop{\mathbb{E}}_{x\leftarrow D_1}\mathop{\mathbb{E}}_{c\leftarrow\mathrm{Enc}(x)}\sum_{\alpha}\bra{\psi}A_{c\alpha}^* \mathcal{B}A_{c\alpha}\ket{\psi}-\mathop{\mathbb{E}}_{x\leftarrow D_2}\mathop{\mathbb{E}}_{c\leftarrow\mathrm{Enc}(x)}\sum_{\alpha}\bra{\psi} A_{c\alpha}^* \mathcal{B}A_{c\alpha}\ket{\psi}\right|\leq \eta'(\lambda).
\end{align*}
\end{lemma}

\begin{proof}
We follow the proof of~\cite[Lemma 15]{nz23}. Suppose for the sake of contradiction that the lemma is false. Then there exists some polynomial~$p(\lambda)$ such that for infinitely many $\lambda$, we have
\begin{align}\label{eq:geqpolyfunct}
    \left|\mathop{\mathbb{E}}_{x\leftarrow D_1}\mathop{\mathbb{E}}_{c\leftarrow\mathrm{Enc}(x)}\sum_{\alpha}\bra{\psi}A_{c\alpha}^* \mathcal{B}A_{c\alpha}\ket{\psi}-\mathop{\mathbb{E}}_{x\leftarrow D_2}\mathop{\mathbb{E}}_{c\leftarrow\mathrm{Enc}(x)}\sum_{\alpha}\bra{\psi} A_{c\alpha}^* \mathcal{B}A_{c\alpha}\ket{\psi}\right|\geq \frac{1}{p(\lambda)}.
\end{align}

Choose $\varepsilon(\lambda) := \frac{1}{100p(\lambda)}$.
% This is possible because $f$ is a polynomial function of $\lambda$.
By \cref{lem:QPTmeasureableuptoepsilon}, there exists a POVM $\{M_{\beta}\}_{\beta}$ such that
\begin{align}\label{eq:leqPOVM}
    \left|\mathop{\mathbb{E}}_{x\leftarrow D_i}\mathop{\mathbb{E}}_{c\leftarrow\mathrm{Enc}(x)}\sum_{\alpha}\sum_{\beta}\beta\bra{\psi}A_{c\alpha}^* M_{\beta}A_{c\alpha}\ket{\psi}-\mathop{\mathbb{E}}_{x\leftarrow D_i}\mathop{\mathbb{E}}_{c\leftarrow\mathrm{Enc}(x)}\sum_{\alpha}\bra{\psi} A_{c\alpha}^* \mathcal{B}A_{c\alpha}\ket{\psi}\right| \leq \varepsilon(\lambda)
\end{align}
for this choice of $\varepsilon(\lambda)$ and $i\in \{1,2\}$. Let $M:=\sum_\beta \beta M_\beta$. From \cref{eq:geqpolyfunct,eq:leqPOVM}, we deduce
\begin{align*}
    \Bigg{|}\mathop{\mathbb{E}}_{x\leftarrow D_1}\mathop{\mathbb{E}}_{c\leftarrow\mathrm{Enc}(x)}\sum_{\alpha}\bra{\psi}A_{c\alpha}^* MA_{c\alpha}\ket{\psi}
    -\mathop{\mathbb{E}}_{x\leftarrow D_2}\mathop{\mathbb{E}}_{c\leftarrow\mathrm{Enc}(x)}\sum_{\alpha}\bra{\psi} A_{c\alpha}^* MA_{c\alpha}\ket{\psi}\Bigg{|}&\geq \frac{1}{p(\lambda)}-2\varepsilon(\lambda)
    \geq \frac{0.98}{p(\lambda)}
\end{align*}
for infinitely many~$\lambda$.
This contradicts \cref{lemma:StrongnonsignallingPOVM}.
\end{proof}

%=============================================================================
\section{Bounding the compiled value of an XOR game}\label{sec:macro}
%=============================================================================
In this section, we will show that for the well-known class of nonlocal games known as XOR games, the quantum value coincides with the quantum value of its compiled version. In particular, we establish the following result:
%This generalizes the result of~\cite{nz23}, where this was shown to be true for the CHSH game, which is a particular XOR game.

\begin{theorem}\label{thm:XORvalue}
Let $\mathcal{G}$ be any XOR game and $\mathcal{G}_\mathrm{comp}$ the corresponding compiled game.
For every strategy $S$ for $\mathcal{G}_\mathrm{comp}$, there exists a negligible function~$\eta(\lambda)$ such that the following holds:
\begin{align*}
\omega_q(S,\mathcal{G}_\mathrm{comp})\leq\omega_q^*(\mathcal{G})+\eta(\lambda).
\end{align*}
\end{theorem}

As a corollary to \cref{thm:XORvalue}, in \cref{sec:self} and \cref{sec:parallel} we derive a self-testing and a parallel repetition theorem for compiled XOR games. Just as we did for ordinary nonlocal games, we define the quantum bias of a compiled game as $\beta_q(S, \mathcal{G}_\mathrm{comp})=2\omega_q(S, \mathcal{G}_\mathrm{comp})-1$ and $\beta_q^*(\mathcal{G}_\mathrm{comp})=2\omega_q^*(\mathcal{G}_\mathrm{comp})-1$.\!\!\!

%Tsirelon's theorem asserts that the optimal bias can be computed by the following optimization over so-called \emph{vector strategies}:

%\begin{theorem}[\cite{Tsirelson}]\label{thm:Tsirelson}
    %Let $\mathcal{G}$ be an XOR game. It holds
    %begin{align}\label{eq:bias}
        %\beta_q^*(\mathcal{G})=\mathop{\mathrm{sup}}_{\substack{u_x,v_y\in %\R^{d}\\\|u_x\|=\|v_y\|=1}} \sum_{x \in I_A,y \in I_B} %G_{xy}\braket{u_x}{v_y}.
 %   \end{align}
%\end{theorem}
%\cnote{I think we should consider moving some of the above to the XOR game section}

%In this section, we provide an alternative proof of \cref{thm:XORvalue} using the SOS approach.

\subsection{SOS's and bounding the value of nonlocal games}\label{subsec:SoS overview}

We briefly review the sum of squares (SOS) approach to bounding the commuting operator value $\omega^*_{qc}(\mcG)$ for a nonlocal game $\mcG$. The starting point is an abstract characterization of a Bell scenario. In particular, we model a two-party (non-communicating) measurement scenario with $n$ measurement settings and $m$ measurement outcomes per party using the abstract \emph{scenario algebra} $\mathcal{A}_{n,m}: = \mathbb{C}[\mathbb{Z}^{*n}_{m} \times \mathbb{Z}^{*n}_{m}]$. From now on we will drop the subscripts, and just write $\mathcal{A}$ for the scenario algebra.\footnote{One can define a similar scenario algebra for when Alice and Bob have different sized input and or output sets.} If we let $a_x$ and $b_y$ to be the order $m$ \emph{abstract observable} generators of $\mathcal{A}$, we observe that there is an alternative characterization of $\mathcal{A}$ in terms of the orthogonal spectral projections $a^{x}_{a}$ and $b^{y}_{b}$, such that
\begin{equation}
    a^{x}_{a} = \frac{1}{m}\sum^{m-1}_{k=0} (\zeta^{-a}a_x)^k , \qquad b^{y}_{b} = \frac{1}{m}\sum^{m-1}_{k=0} (\zeta^{-b}b_y)^k.
\end{equation}

Where $\zeta$ is an $m$th primitive root of unity. We briefly recall some more definitions. An element $y\in \mathcal{A}$ is a (hermitian) square if $y=x^*x$ for some $x\in \mathcal{A}$. Moreover, we say an element of $\mathcal{A}$ is a sum-of-squares (SOS) if it can be written as a sum of (hermitian) squares in $\mathcal{A}$. A linear functional $f: \mathcal{A} \rightarrow \mathbb{C}$ is said to be \emph{positive} if $f(x^*x) \geq 0$ for all $x \in \mathcal{A}$. In this setting,\footnote{Here we are viewing the group algebra $\mathbb{C}[\mathbb{Z}^{*n}_{m} \times \mathbb{Z}^{*n}_{m}]$ as a semi-pre-$C^*$-algebra, with the positive cone generated by sums of squares as described in \cite{Ozawa2013}. The corresponding universal $C^*$ algebra is canonically isomorphic to $C^{*}(\mathbb{Z}^{*n}_{m}) \otimes_{max} C^{*}(\mathbb{Z}^{*n}_{m})$.} an \emph{expectation} on $\mathcal{A}$ (also called an abstract state) is a linear functional $\E: \mathcal{A} \rightarrow \mathbb{C}$, that is both positive and unital (i.e  $\E(\Id)=1$). Every commuting operator strategy $\mathcal{S}=(\ket{\psi}, \lbrace A_{xa} \rbrace, \lbrace B_{yb} \rbrace )$ determines an expectation $\mathbb{E}_{\mathcal{S}} : \mathcal{A} \rightarrow \mathbb{C}$. The expectation $\mathbb{E}_{\mathcal{S}}$ is determined by first applying the $*$-representation mapping $a_{xa} \mapsto A_{xa}$ and $b^{y}_{b} \mapsto B_{yb}$, where the inner-product of the resulting operators with the state $\ket{\psi}\in \mathcal{H}$ is then taken in the appropriate Hilbert space. Conversely, the GNS theorem for $*$-algebras \cite{Schmudgen2020Unbounded, Ozawa2013} shows that every expectation $\E:\mathcal{A} \rightarrow \mathbb{C}$ gives rise to a commuting operator strategy $\mathcal{S}$.

Given a nonlocal game $\mcG$, the associated \emph{game polynomial} \begin{equation}
    p_{\mcG}: = \sum_{a,b,x,y} \pi(x,y)V(x,y,a,b) a^{x}_{a}b^{y}_{b},
\end{equation}
is an element of the scenario algebra $\mathcal{A}$. Hence, we can realize the commuting operator value $\omega^{*}_{qc}(\mcG)$ as the supremum of $p_{\mcG}$ over all expectations on $\mathcal{A}$. Furthermore, if $\nu \Id-p_{\mcG}$ is expressible as an SOS then $\nu \geq \omega^*_{qc}(\mcG)$. For all strict upper bounds $\nu > \omega^*_{qc}(\mcG)$ one can find an SOS decomposition for $\nu\Id - p_{\mcG}$ \cite{Ozawa2013}. For the exact value \(\nu=\omega^*_{qc}(\mathcal{G})\), it remains uncertain whether an SOS decomposition always exists \cite{mehta2023positivity}.

For some families of nonlocal games, it is more convenient to consider the \emph{bias polynomial} $h_{\mcG}$, rather than the game polynomial. For instance, in the binary output case, the bias polynomial is expressed in terms of the order 2 observables $a_x$ and $b_y$. The supremum over expectations of $h_{\mcG}$ is the commuting operator bias $\beta^*_{qc}(\mcG)$, and obtaining a sum of squares decomposition for $\beta\Id -h_{\mcG}$ certifies that $\beta$ as an upper bound $\beta \geq \beta^*_{qc} \geq \beta^*_{q}$. An instance of this will be seen in \cref{subsec:NiceSOS}, where we describe the bias polynomial of an XOR game $\mcG$ and find associated SOS decomposition's of $\beta\Id -h_{\mcG}$.

\subsection{Bounding the compiled value using cryptographic pseudo-expectations}

The SOS method for establishing an upper-bound on $\omega^{*}_{qc}(\mcG)$ relies on the correspondence between commuting operator strategies $\mathcal{S}$ and expectations $\mathbb{E}_{\mathcal{S}}$ on the scenario algebra $\mathcal{A}$. Despite not having a correspondence between quantum strategies for compiled games $\mcG_{\mathrm{comp}}$ and expectations, we will see that certain SOS decompositions for $h_{\mcG}$ can be used to bound the compiled value.

To illustrate this approach suppose we have a sum-of-squares decomposition $\beta\Id - h_{\mcG} = \sum_{i\in \mathcal{I}} r_{i}^{\ast}r_{i}$ and let $\mathcal{T} \subset \mathcal{A}$ denote the subspace of $\mathcal{A}$ that contains $\Id$ and $r_{i}^{\ast}r_{i}$ for all $i\in \mathcal{I}$. Furthermore, suppose we have a function $\tilde{\E} : \mathcal{T} \to \C$,  with the following properties:
\begin{enumerate}[(1)]
    \item $\tilde{\E}[\cdot]$ is linear,
    \item $\tilde{\E}[\Id] = 1$,
    \item $\tilde{\E}[h_{\mathcal{G}}] = \beta_{q}(S, \mathcal{G}_{\mathrm{comp}})$, and
    \item For all $i\in \mathcal{I}$, $\tilde{\E}[r_i^{\ast}r_i] \geq -\delta$, for some $\delta\geq 0$.
\end{enumerate}
Then,
\[ \beta - \beta_{q}(S, \mathcal{G}_{\mathrm{comp}}) = \tilde{\E}\sbrac{ \beta \Id - h_{\mathcal{G}} } = \tilde{\E}\sbrac{ \sum_{i} r_{i}^{\ast}r_{i} } = \sum_{i}\tilde{\E}[ r_{i}^{\ast}r_{i}] \geq -| \mathcal{I}|\delta, \]
giving us an upper-bound for $\beta_{q}(S, \mathcal{G}_{\mathrm{comp}})$.

For every strategy for the compiled game $S$, we can define a ``pseudo-expectation'' $\tilde{\E}_S:\mathcal{T}\rightarrow \mathbb{C}$ that satisfies properties (1)-(3) and (4) with negligible $\delta$. In \cref{subsec:NiceSOS}, we use this approach to bound the compiled value of all XOR games.

\begin{definition}
    Let $S = (\ket{\psi},\{A_{c\alpha}\}, \{B_{yb}\})$ be a quantum strategy for the compiled game. The \emph{(degree 2) pseudo-expectation induced by $S$}, is the function $\tilde{\E}_{S}$ defined  on all degree $2$ monomials in $\mathcal{A}$ by
    \begin{align}
        \tilde{\mathbb{E}}_{S}\sbrac{a_{x}a_{z}} &= \begin{cases} 1, \text{ if $x=z$}\\ 0, \text{ otherwise} \end{cases} \\
        \tilde{\mathbb{E}}_{S}\sbrac{a_{x}b_{y}} = \tilde{\E}_{S}\sbrac{b_{y}a_{x}} &= \E_{c\leftarrow\Enc(x)}\sum_\alpha(-1)^{\Dec(\alpha)}\langle\psi|(A_{c\alpha})^*B_{y}A_{c\alpha}|\psi\rangle \\
        \tilde{\E}_{S}\sbrac{b_{y}b_{w}} &= \E_{x}\E_{c\leftarrow\Enc(x)}\sum_\alpha\langle\psi|(A_{c\alpha})^*B_{y}B_{w} A_{c\alpha}|\psi\rangle, \label{eqn:bterms}
    \end{align}
    for all $x,z\in I_A$ and $y,w\in I_B$.
\end{definition}

\begin{remark}
    We note that the distribution over $x$ in \cref{eqn:bterms} can be replaced with any efficiently sampleable distribution and the results below would still hold.
\end{remark}

From the definition, we immediately obtain that $\tilde{\E}_{S}[\cdot]$ is linear and unital. We next show that it is non-negative up to some negligible function on certain squares terms.

\begin{theorem}\label{theorem:positiveSoS}
    Let $\tilde{\E}_S[\cdot]$ be the degree 2 pseudo-expectation induced by a strategy for the compiled game $S$. Then
    \[ \tilde{\E}_{S}\sbrac{ \brac{\sum_{y} \gamma_{y} b_{y}}^{2} } \geq 0, \]
    where $\gamma_{y} \in \R$,
    and there exists a negligible function of the security parameter $\lambda$ such that
    \[ \tilde{\E}_{S}\sbrac{ \brac{ a_{x} - \sum_{y} \gamma_{x,y} b_{y} }^{2} } \geq -\eta(\lambda), \]
    where $\gamma_{x,y} \in \R$. Note that the order of the quantifiers imply the negligible function $\eta$ may depend on $S$.
\end{theorem}

\begin{proof}
    Let us fix a strategy for the compiled game $S = (\ket{\psi},\{A_{c\alpha}\}, \{B_{yb}\})$. To show the non-negativity of the first family of square terms we write
    \begin{align}
        \tilde{\E}\sbrac{ \brac{ \sum_{y} \gamma_{y} b_{y}
 }^{2} } &= \sum_{y, w} \gamma_{y}\gamma_{w} \tilde{\E}\sbrac{ b_{y}b_{w} } \\
        &= \sum_{y, w} \gamma_{y}\gamma_{w} \E_{x}\E_{c\leftarrow\mathrm{Enc}(x)} \sum_{\alpha}\bra{\psi}(A_{c\alpha})^{\ast}B_{y}B_{w}A_{c\alpha}\ket{\psi} \\
        &= \E_{x}\E_{c\leftarrow\mathrm{Enc}(x)} \sum_{\alpha}\bra{\psi}(A_{c\alpha})^{\ast}\brac{ \sum_{y, w} \gamma_{y}\gamma_{w} B_{y}B_{w} }A_{c\alpha}\ket{\psi} \\
        &= \E_{x}\E_{c\leftarrow\mathrm{Enc}(x)} \sum_{\alpha}\bra{\psi}(A_{c\alpha})^{\ast}\brac{ \sum_{y} \gamma_{y}B_{y} }^{\ast}\brac{ \sum_{y} \gamma_{y}B_{y} } A_{c\alpha}\ket{\psi} \\
        &= \E_{x}\E_{c\leftarrow\mathrm{Enc}(x)} \sum_{\alpha}\bra{\psi}\brac{ \brac{ \sum_{y} \gamma_{y}B_{y} } A_{c\alpha} }^{\ast} \brac{  \brac{ \sum_{y} \gamma_{y}B_{y} } A_{c\alpha} }\ket{\psi} \\
        &= \E_{x}\E_{c\leftarrow\Enc(x)} \sum_{\alpha} \left\lVert
\brac{ \sum_{y} \gamma_{y} B_{y}} \ket{\psi_{c\alpha}} \right\rVert^{2} \geq 0.
    \end{align}
    The proof for the second type of square terms depends on the security of the QHE scheme. Let $\mu_{x}$ be the distribution on $(a,b) \in \R^{2}$ generated by the following procedure:
    \begin{enumerate}
        \item Sample a uniformly random encoding $c \leftarrow \Enc(x)$.
        \item Measure $\ket{\psi}$ with the observable $A_{c}$, obtaining outcome $\alpha$ which is then decrypted as $a= (-1)^{\Dec(\alpha)}$.
        \item Measure the post-measurement state with the observable $\hat{B}_{x} := \sum_{y}\gamma_{x,y}B_{y}$, obtaining some outcome $b \in \R$.
    \end{enumerate}

    Let $\hat{B}_{x} = \sum_{b} b\hat{B}_{xb}$ be the spectral decomposition of $\hat{B}_{x}$. The expectation of $(a-b)^2$ over $\mu_x$ is then given by
    \begin{align*}
        \E_{\mu_{x}}\sbrac{ (a-b)^{2} } &= \E_{c\leftarrow\Enc(x)} \sum_{\alpha}\sum_{b} \bra{\psi}(A_{c\alpha})^{\ast}\hat{B}_{xb}A_{c\alpha}\ket{\psi}(a-b)^{2} \\
            &= \E_{c\leftarrow\Enc(x)} \sum_{\alpha}\sum_{b} \bra{\psi}(A_{c\alpha})^{\ast}\hat{B}_{xb}A_{c\alpha}\ket{\psi}(1 - 2(-1)^{\Dec(\alpha)}b + b^{2}) \\
            &= 1 - 2\E_{c\leftarrow\Enc(x)} \sum_{\alpha,b} \bra{\psi} (A_{c\alpha})^{\ast}\hat{B}_{xb}A_{c\alpha}\ket{\psi} (-1)^{\Dec(\alpha)}b + \E_{c\leftarrow\Enc(x)} \sum_{\alpha,b} \bra{\psi}(A_{c\alpha})^{\ast}\hat{B}_{xb}A_{c\alpha}\ket{\psi}b^{2} \\
            &= 1 - 2 \E_{c\leftarrow\Enc(x)} \sum_{\alpha} (-1)^{\Dec(\alpha)}\bra{\psi} (A_{c\alpha})^{\ast}\hat{B}_{x}A_{c\alpha}\ket{\psi} + \E_{c\leftarrow\Enc(x)} \sum_{\alpha} \bra{\psi}(A_{c\alpha})^{\ast}\hat{B}_{x}^{2}A_{c\alpha}\ket{\psi}.
    \end{align*}
    Note that $\hat{B}_{x}$ and $\hat{B}_{x}^{2}$ have QPT-implementable block encodings with scale factor~$\Theta(1)$ by \cref{lem:block from binary,lem:block lincom,lem:block prod}. So by \cref{lemma:Strongnonsignalling}, there exists a negligible function $\eta(\lambda)$ such that
    \begin{equation}\label{eq:swap_dist}
    \left|\E_{c\leftarrow\Enc(x)}\sum_\alpha \langle \psi|(A_{c\alpha})^*\hat{B}_{x}^2 A_{ c\alpha}|\psi\rangle -\E_{x'}\E_{c'\leftarrow\Enc(x')}\sum_\alpha \langle\psi|(A_{c' \alpha})^*\hat{B}_{x}^2A_{c' \alpha}|\psi\rangle\right|\leq \eta(\lambda).
    \end{equation}

    Now, if we expand
    \begin{align*}
        \tilde{\E}\sbrac{ \brac{ a_{x} - \sum_{y} \gamma_{x,y} b_{y} }^{2} } &= \tilde{\E}\sbrac{ a_{x}^{2} } - 2 \sum_{y} \gamma_{x,y}\tilde{\E}\sbrac{ a_{x}b_{y} } + \sum_{y,w}\gamma_{x,y}\gamma_{x,w}\tilde{\E}\sbrac{ b_{y}b_{w} } \\
            &= 1 - 2 \sum_{y}\gamma_{x,y}\E_{c\leftarrow\Enc(x)}\sum_\alpha(-1)^{\Dec(\alpha)}\langle\psi|(A_{c\alpha})^*B_{y}A_{c\alpha}|\psi\rangle \\ &\qquad\qquad\qquad+ \sum_{y,w}\gamma_{xy}\gamma_{xw} \E_{x'}\E_{c'\leftarrow\Enc(x')}\sum_\alpha\langle\psi|(A_{c'\alpha})^*B_{y}B_{w} A_{c'\alpha}|\psi\rangle \\
            &= 1 - 2 \E_{c\leftarrow\Enc(x)}\sum_{\alpha}(-1)^{\Dec(a)}\bra{\psi}(A_{c\alpha})^{\ast}\brac{ \sum_{y}\gamma_{x,y}B_{y} }A_{c\alpha} \ket{\psi} \\
                &\qquad\qquad\qquad + \E_{x'}\E_{c'\leftarrow\Enc(x')}\sum_{\alpha}\bra{\psi}(A_{c'\alpha})^{\ast}\brac{ \sum_{y}\gamma_{xy}B_{y} }^{2}A_{c'\alpha}\ket{\psi} \\
            &= 1 - 2 \E_{c\leftarrow\Enc(x)}\sum_{\alpha}(-1)^{\Dec(a)}\bra{\psi}(A_{c\alpha})^{\ast}\hat{B}_{x}A_{c\alpha} \ket{\psi} \\& \qquad\qquad\qquad+ \E_{x'}\E_{c\leftarrow\Enc(x')}\sum_{\alpha}\bra{\psi}(A_{c'\alpha})^{\ast}\hat{B}_{x}^{2}A_{c'\alpha}\ket{\psi} \\
            &\geq \E_{\mu_{x}}\sbrac{(a-b)^{2}} - \eta(\lambda) \\
            &\geq -\eta(\lambda),
    \end{align*}
    where the second to last inequality follows from \cref{eq:swap_dist}, completing the proof.
\end{proof}

\begin{definition}
    We call a sum-of-squares decomposition \emph{nice} if it is a non-negative linear combination of terms of the form given in \cref{theorem:positiveSoS}.
\end{definition}

\begin{corollary} \label{theorem:niceSoStoUpperbound}
    Suppose that nonlocal game $\mathcal{G}$ has a nice sum-of-squares decomposition for $\beta \Id - h_{\mathcal{G}}$, then for any strategy for the compiled game $S$ there exists a negligible function $\eta(\lambda)$ of the security parameter such that
    \[ \beta_{q}(S,\mathcal{G}_{\mathrm{comp}}) \leq \beta + \eta(\lambda). \]
\end{corollary}

\subsection{A nice SOS certificate for any XOR game}\label{subsec:NiceSOS}

In this section, we exhibit a nice SOS decomposition for $\beta^{\ast}_{q}(\mathcal{G}) \Id - h_{\mathcal{G}}$ for every XOR game $\mathcal{G}$, where $\beta^{\ast}_{q}(\mathcal{G})$ is the optimal quantum bias for $\mcG$ and $h_{\mathcal{G}}$ is the bias polynomial associated with $\mathcal{G}$. This, in a sense, is the dual view of Tsirelson's theorem.

Recall that an XOR game $\mathcal{G}$ is described by its cost matrix $G = (G_{xy})_{x \in I_{A}, y \in I_{B}}$. The bias polynomial is expressed as
\[ h_{\mathcal{G}} = \sum_{x \in I_{A}, y \in I_{B}} G_{xy}a_{x}b_{y} \in \mathcal{A}. \]

\begin{theorem}
    For any XOR game $\mathcal{G}$ and any quantum strategy $S$ for the compiled game,
    \[ \tilde{\mathbb{E}}_{S}[h_{\mathcal{G}}] = \beta_{q}(S, \mathcal{G}_{\mathrm{comp}}). \]
\end{theorem}
\begin{proof}
    Fix a quantum strategy $S = (\ket{\psi}, \cbrac{A_{c\alpha}}, \cbrac{B_{yb}})$, then
    \begin{align*}
        \tilde{\E}_{S}\sbrac{ h_{\mathcal{G}} } &= \sum_{x \in I_{A}, y \in I_{B}} G_{xy} \tilde{\E}_{S}\sbrac{a_{x}b_{y}} \\
            &= \sum_{x \in I_{A}, y \in I_{B}} G_{xy} \E_{c \leftarrow \Enc(x)} \sum_{\alpha} (-1)^{\Dec(\alpha)} \bra{\psi}(A_{c\alpha})^{\ast}B_{y}A_{c\alpha}\ket{\psi} \\
            &= \sum_{x \in I_{A}, y \in I_{B}} G_{xy} \sum_{a, b} \E_{c \leftarrow \Enc(x)} \sum_{\alpha; \Dec(\alpha) = a} (-1)^{a+b} \bra{\psi}(A_{c\alpha})^{\ast}B_{yb}A_{c\alpha}\ket{\psi} \\
            &= \sum_{x\in I_{A}, y\in I_{B}} \pi(x,y)  \sum_{a, b} (-1)^{g(x,y)}(-1)^{a+b} \E_{c \leftarrow \Enc(x)} \sum_{\alpha; \Dec(\alpha) = a}  \bra{\psi}(A_{c\alpha})^{\ast}B_{yb}A_{c\alpha}\ket{\psi} \\
            &= 2 \omega_{q}(S, \mathcal{G}_{\mathrm{comp}}) - 1.
    \end{align*}
\end{proof}

Let $\beta^{\ast}_{q}(\mathcal{G})$ be the optimal quantum bias for $\mathcal{G}$. By \cref{thm:Tsirelson} for any optimal quantum strategy $S$ for the XOR game $\mcG$ the corresponding row and column biases $r_{x,S}$ and $c_{y,S}$ will be the same, which we simply denote as $r_{x}$ and $c_{y}$.
For each $x \in I_{A}$ we then define the following element of the scenario algebra, \begin{equation*}
    \hat{b}_{x} := \frac{1}{r_{x}}\sum_{y \in I_{B}} G_{xy}b_{y}.
\end{equation*} We begin by expanding the following expression
\begin{align*}
    \sum_{x \in I_{A}} \frac{r_{x}}{2} \brac{ a_{x} - \hat{b}_{x} }^{2} &= \frac{1}{2}\beta \Id + \frac{1}{2} \sum_{x \in I_{A}} r_{x} \hat{b}_{x}^{2} - \sum_{x \in I_{A}} r_{x} a_{x} \hat{b}_{x} \\
        &= \beta \Id + \frac{1}{2} \sum_{x \in I_{A}}r_{x}\brac{ \hat{b}_{x}^{2} - \Id } - \sum_{x\in I_{A}, y \in I_{B}} G_{xy}a_{x}b_{y} \\
        &= \beta \Id + \frac{1}{2} \sum_{x \in I_{A}} r_{x}\brac{ \hat{b}_{x}^{2} - \Id } - h_{\mathcal{G}}.
\end{align*}
If we show that the negative of $\frac{1}{2} \sum_{x \in I_{A}} r_{x}\brac{ \hat{b}_{x}^{2} - \Id }$ has a sum-of-squares decomposition, then we would be done. Towards this,
\begin{equation}
    - \sum_{x\in I_{A}} r_{x}\brac{ \hat{b}_{x}^{2} - \Id } = \beta \Id - \sum_{x\in I_{A}}r_{x}\hat{b}_{x}^{2} = \sum_{y \in I_{B}} c_{y}b_{y}^{2} - \sum_{x\in I_{A}} r_{x}\hat{b}_{x}^{2}
\end{equation}
Now, let $\mathbf{b}$ and $\mathbf{\hat{b}}$ be the column vectors of $\{b_{y}:y \in I_{B}\}$ and $\{\hat{b}_{x}:x\in I_{A}\}$, respectively. Similarly, we let $\mathbf{r}$ and $\mathbf{c}$ denote the vectors of $\cbrac{r_{x}}$ and $\cbrac{c_{y}}$, respectively. With this, we can write
\begin{align*}
    \sum_{x \in I_{A}} r_{x}\hat{b}_{x}^{2} &= \mathbf{\hat{b}}^{\top} \Delta(\mathbf{r}) \mathbf{\hat{b}} \\
        &= \brac{ \Delta(\mathbf{r})^{-1} G \mathbf{b} }^{\top} \Delta(\mathbf{r}) \brac{ \Delta(\mathbf{r})^{-1}G \mathbf{b} } \\
        &= \mathbf{b}^{\top} G^{\top} \Delta(\mathbf{r})^{-1} G \mathbf{b}.
\end{align*}
Putting this all together we obtain,
\[ \sum_{y \in I_{B}} c_{y} b_{y}^{2} - \sum_{x\in I_{A}} r_{x} \hat{b}_{x}^{2} = \mathbf{b}^{\top}\Delta(\mathbf{c})\mathbf{b} - \mathbf{b}^{\top} G^{\top} \Delta(\mathbf{r})^{-1} G \mathbf{b} = \mathbf{b}^{\top}\brac{ \Delta(\mathbf{c}) - G^{\top} \Delta(\mathbf{r})^{-1} G }\mathbf{b}. \]
By \cref{thm:slofstra},
\[ \begin{pmatrix} \Delta(\mathbf{r}) & -G^{\top} \\ -G & \Delta(\mathbf{c}) \end{pmatrix} \succeq 0. \]
In addition, the matrix $\Delta(\mathbf{r})$ is invertible since each $r_{x}$ is strictly positive. Hence, by the Schur complement \cite[Theorem 1.12 (b)]{hornzhang} condition for positivity, we obtain
\[ M := \Delta(\mathbf{c}) - G^{\top}\Delta(\mathbf{r})^{-1}G \succeq 0. \]
Lastly, taking the spectral decomposition $M = \sum_{y \in I_{B}} \lambda_{y} \mathbf{v}_{y}\mathbf{v}_{y}^{\ast}$ gives us the sum-of-squares decomposition
\[ \sum_{y \in I_{B}} c_{y}b_{y}^{2} - \sum_{x\in I_{A}} r_{x}\hat{b}_{x}^{2} = \sum_{y \in I_{B}} \lambda_{y} \brac{ \sum_{w \in I_{B}} v_{yw}b_{w} }^{2}, \]
where $\mathbf{b}^{\top} \mathbf{v}_{y}\mathbf{v}_{y}^{\ast}\mathbf{b}=\left(\sum_w v_{yw}b_w\right)^2$, for each $y\in I_B$. We summarize with the following theorem:

\begin{theorem} \label{theorem:xorSoS}
For any XOR game $\mcG$, there is a sum-of-squares certificate for the optimal quantum bias $\beta$ of the form
\begin{align}
    \beta \Id - h_{\mathcal{G}} = \sum_{x \in I_{A}} \frac{r_{x}}{2} \brac{a_{x} - \hat{b}_{x}}^2 + \sum_{y \in I_{B}} \frac{\lambda_{y}}{2} \brac{ \sum_{w \in I_{B}} v_{yw} b_{w} }^{2},
\end{align}
where each $v_{yw}\in \mathbb{R}$ and $\lambda_{y}$ are non-negative for all $y,w\in I_B$.
\end{theorem}

\Cref{theorem:xorSoS} allows us to upper-bound any strategy for the compiled game for any XOR game in terms of the optimal quantum bias and the security parameter of the compilation scheme. Hence, we obtain our second proof of \cref{thm:XORvalue}, stated in \cref{sec:macro}.

%\begin{theorem}\label{thm:XORvalueagain}
 %   For XOR game $\mathcal{G}$ and quantum strategy for the compiled game $S$, there exists some negligible function $\eta(\lambda)$ of the security parameter such that
 %   \[ \beta_{q}(S, \mathcal{G}_{\mathrm{comp}}) \leq \beta^{\ast}_{q}(\mathcal{G}) + \eta(\lambda). \]
%\end{theorem}
\begin{proof}[Proof of \cref{thm:XORvalue}]
    From \cref{theorem:niceSoStoUpperbound} and \cref{theorem:xorSoS}, we obtain the desired result.
\end{proof}

%Likely one can obtain the same self-testing result as proved earlier (in an even easier way).

%We mention that the SOS approach outlined here for XOR games likely can be extended to find an upper bound for the compiled value for the family of parallel repeated XOR games discussed in Section \ref{sec:parallel}. In particular, one can obtain a version of Lemma \ref{lem:parallel_to_sum} on the level of abstract game and bias polynomials. Then, because the direct sum of XOR games is an XOR game, we have a nice SOS decomposition by Theorem \ref{theorem:xorSoS} for each bias polynomial associated with each direct sum $\oplus_{M\subset [n]}\mcG_i$ game. Thus, we can apply Theorem \ref{theorem:positiveSoS} to each bias polynomial $h_\mcG^M$ corresponding to the subsets $M\subseteq[n]$. This approach likely provides the same bound obtained in Theorem \ref{thm:XORparallel}, so we omit the details.

%=============================================================================

Together, \cref{thm:KLVY} and \cref{thm:XORvalue} imply that the ``asymptotic'' quantum value of a compiled XOR game is equal to the quantum value of the underlying XOR game.
That is, if~$S$ is any quantum strategy for the compiled game, then
\begin{align*}
    \limsup_{\lambda\to\infty} \omega_q(S,\mathcal{G}_\mathrm{comp}) \leq w_q^*(\mathcal G),
\end{align*}
and conversely, there always exists a strategy such that
\begin{align*}
    \liminf_{\lambda\to\infty} \omega_q(S,\mathcal{G}_\mathrm{comp}) \geq w_q^*(\mathcal G).
\end{align*}

%-----------------------------------------------------------------------------
\section{Self-testing results for compiled XOR games}\label{sec:self}
%-----------------------------------------------------------------------------
We now give a self-testing type result which allows us to extract partial information about a quantum strategy if the winning probability is near-optimal.

%We first discuss the situation for the nonlocal game.
Let $\mathcal{G}$ be an XOR game with cost matrix~$G$.
Recall there exist \emph{row biases}~$r_x\geq 0$ for every $x\in I_A$ such that $r_x=\sum_y G_{xy}\braket{u_x}{v_y}$ for every optimal vector strategy~$\{u_x, v_y\}$.
Slofstra showed that these row (and column) biases are robust in the sense that for an $\varepsilon$-optimal vector strategy, it holds that
\begin{align}\label{eq:almostoptimalvectorstrat}
   \left(r_x- \sum_y G_{xy}\braket{u_x}{v_y}\right)^2 \leq 2(|I_A|+|I_B|)\beta_q^*(\mathcal{G})\varepsilon,
\end{align}
see~\cite[Section 3]{Slofstra11}.
Note that the quantity~$\varepsilon$ refers to the winning probability, whereas in~\cite[Section 3]{Slofstra11}, the notion of $\varepsilon$-optimality is defined with respect to the bias.

Slofstra furthermore defined the \emph{solution algebra} $\mathcal{A}(\mathcal{G})$ to be the universal $C^*$-algebra generated by self-adjoint elements~$Y_y$ for $y\in I_B$, subject to the relations:
\begin{align*}
    Y_y^2=\Id \text{ for all } y \in I_B \qquad\text{and}\qquad \left(\sum_{y \in I_B} G_{xy}Y_y\right)^2=r_x^2\Id \text{ for all } x\in I_A.
\end{align*}

We will see that in an $\varepsilon$-optimal quantum strategy of the compiled XOR game, the Bob operators fulfill similar relations (\cref{thm:XORselftest}).
To this end we make the following definition, which is justified by the discussion at the end of the preceding section.

\begin{definition}
Let $\mathcal{G}$ be an XOR game and $\mathcal{G}_\mathrm{comp}$ the corresponding compiled game.
We say that a quantum strategy~$S$ for~$\mathcal{G}_\mathrm{comp}$ is \emph{$\eps$-optimal} if (for sufficiently large~$\lambda$)
\begin{align*}
    \omega_q(S, \mathcal{G}_\mathrm{comp}) \geq \omega_q^*(\mathcal G) - \eps.
\end{align*}
\end{definition}

\begin{lemma}\label{lem:XORselftest}
Let $\mathcal{G}$ be an XOR game with cost matrix $G$. If $S=(\ket{\psi}, \{A_{c\alpha}\}, \{B_{yb}\})$ is an $\varepsilon$-optimal quantum strategy for $\mathcal{G}_\mathrm{comp}$.
Then there exists a negligible function~$\eta'(\lambda)$ such that for every~$z\in I_A$ it holds that
\begin{align*}
    \mathop{\mathbb{E}}_{c\leftarrow\mathrm{Enc}(z)}\sum_{\alpha}\| \parens*{ r_z\Id-(-1)^{\mathrm{Dec}(\alpha)}\sum_y G_{zy}B_y }\ket{\psi_{c\alpha}}\|^2
\leq 2 \parens[\big]{ \parens*{ \abs{I_A}+\abs{I_B} } \, \beta_q^*(\mathcal{G})+1 } \varepsilon',
\end{align*}
where $\varepsilon'= 2\varepsilon + \eta'(\lambda)$.
\end{lemma}

\begin{proof}
Since $S$ is $\varepsilon$-optimal, and using \cref{eq:proofXOR,eq:proofXOR3} in the proof of \cref{thm:XORvalue} (with~$x_0 = z$), we see that
\begin{align}
\nonumber
\beta_q^*(\mathcal G) - 2\eps
&\leq \beta_q^*(S, \mathcal G_\mathrm{comp}) \\
\nonumber
&= \sum_{x \in I_A} \mathop{\mathbb{E}}_{c\leftarrow\mathrm{Enc}(x)}\sum_{\alpha}(-1)^{\mathrm{Dec}(\alpha)}\bra{\psi_{c\alpha}} \left(\sum_yG_{xy}B_y\right)\ket{\psi_{c\alpha}} \\
\nonumber
&\leq \sum_{x \in I_A} \sqrt{\mathop{\mathbb{E}}_{c\leftarrow\mathrm{Enc}(x)}\sum_{\alpha}\bra{\psi_{c\alpha}} \left(\sum_yG_{xy}B_y\right)^2\ket{\psi_{c\alpha}}} \\
\nonumber
&\leq \sum_{x \in I_A, y \in I_B} G_{xy} \braket{u_x}{v_y} + \eta'(\lambda) \\
\label{eq:chain}
&\leq \beta^*_q(\mathcal G) + \eta'(\lambda),
\end{align}
where~$\{u_x,v_y\}$ is a vector strategy (depending on the choice of~$z$) such that, for all~$x\in I_A$,
\begin{align}\label{eq:vector for z}
  \sum_{y \in I_B} G_{xy} \braket{u_x}{v_y}
= \sqrt{\mathop{\mathbb{E}}_{c\leftarrow\mathrm{Enc}(z)}\sum_{\alpha}\bra{\psi_{c\alpha}} \left(\sum_yG_{xy}B_y\right)^2\ket{\psi_{c\alpha}}}.
\end{align}
From this we can draw two conclusions:

First, since the second inequality in \cref{eq:chain} holds for each~$x\in I_A$ (by \cref{lemma:Jensen}), we obtain that
\begin{align*}
    \sqrt{\mathop{\mathbb{E}}_{c\leftarrow\mathrm{Enc}(x)}\sum_{\alpha}\bra{\psi_{c\alpha}} \left(\sum_yG_{xy}B_y\right)^2\ket{\psi_{c\alpha}}}
    \quad- \mathop{\mathbb{E}}_{c\leftarrow\mathrm{Enc}(x)}\sum_{\alpha}(-1)^{\mathrm{Dec}(\alpha)}\bra{\psi_{c\alpha}} \left(\sum_yG_{xy}B_y\right)\ket{\psi_{c\alpha}}
\leq \eps',
\end{align*}
where $\eps' := 2\eps + \eta'(\lambda)$.

From this, it follows that
\begin{align}
\nonumber
&\quad \mathop{\mathbb{E}}_{c\leftarrow\mathrm{Enc}(x)}\sum_{\alpha}\| \parens*{ r_x-(-1)^{\mathrm{Dec}(\alpha)}\sum_y G_{xy}B_y } \ket{\psi_{c\alpha}}\|^2 \\
% &=\mathop{\mathbb{E}}_{c\leftarrow\mathrm{Enc}(x)}\sum_{\alpha}\bra{\psi_{c\alpha}}r_x^2-2r_x(-1)^{\mathrm{Dec}(\alpha)}\sum_y G_{xy}B_y+\left(\sum_y G_{xy}B_y\right)^2\ket{\psi_{c\alpha}} \\
\nonumber
&= r_x^2
   - 2 r_x \mathop{\mathbb{E}}_{c\leftarrow\mathrm{Enc}(x)} \sum_{\alpha} (-1)^{\mathrm{Dec}(\alpha)} \bra{\psi_{c\alpha}} \parens*{ \sum_y G_{xy}B_y } \ket{\psi_{c\alpha}}
   + \mathop{\mathbb{E}}_{c\leftarrow\mathrm{Enc}(x)} \sum_{\alpha} \bra{\psi_{c\alpha}} \left(\sum_y G_{xy}B_y\right)^2 \ket{\psi_{c\alpha}} \\
\nonumber
&\leq r_x^2
   - 2 r_x \sqrt{\mathop{\mathbb{E}}_{c\leftarrow\mathrm{Enc}(x)}\sum_{\alpha}\bra{\psi_{c\alpha}} \left(\sum_yG_{xy}B_y\right)^2\ket{\psi_{c\alpha}}}
   + \mathop{\mathbb{E}}_{c\leftarrow\mathrm{Enc}(x)} \sum_{\alpha} \bra{\psi_{c\alpha}} \left(\sum_y G_{xy}B_y\right)^2 \ket{\psi_{c\alpha}}
   + 2 r_x \eps' \\
\label{eq:proofXORineq2}
&\leq \left(r_x-\sqrt{\mathop{\mathbb{E}}_{c\leftarrow\mathrm{Enc}(x)}\sum_{\alpha}\bra{\psi_{c\alpha}} \left(\sum_yG_{xy}B_y\right)^2\ket{\psi_{c\alpha}}}\right)^2+2\varepsilon',
\end{align}
where we used that~$r_x \leq 1$ in the last step.

Second, from \cref{eq:chain} we see that the vector strategy~$\{u_x,v_y\}$ is $\eps'$-optimal.
Thus it follows from \cref{eq:vector for z,eq:almostoptimalvectorstrat} that
  \begin{align*}
  \parens*{ r_x - \sqrt{\mathop{\mathbb{E}}_{c\leftarrow\mathrm{Enc}(z)}\sum_{\alpha}\bra{\psi_{c\alpha}} \left(\sum_yG_{xy}B_y\right)^2\ket{\psi_{c\alpha}}} }^2
= \parens*{ r_x - \sum_y G_{xy}\braket{u_x}{v_y} }^2
 \leq 2(|I_A|+|I_B|)\beta_q^*(\mathcal{G})\varepsilon',
\end{align*}
for all~$x \in I_A$. The lemma follows from this and \cref{eq:proofXORineq2} for~$x=z$.
\end{proof}

The next theorem shows that in a near-optimal strategy, Bob's operators nearly satisfy the relations of the solution algebra $\mathcal{A}(\mathcal{G})$ with respect to a certain expectation and state-dependent norm.

\begin{theorem}\label{thm:XORselftest}
Let $\mathcal{G}$ be an XOR game with cost matrix $G$. If $S=(\ket{\psi}, \{A_{c\alpha}\}, \{B_{yb}\})$ is an $\varepsilon$-optimal quantum strategy for $\mathcal{G}_\mathrm{comp}$,
then there exists a negligible function~$\eta'(\lambda)$ such that for every~$x\in I_A$
\begin{align*}
\mathop{\mathbb{E}}_{c\leftarrow\mathrm{Enc}(x)}\sum_{\alpha}\|\left(r_x^2-\left(\sum_y G_{xy}B_y\right)^2\right)\ket{\psi_{c\alpha}}\|^2\leq C_1\varepsilon',
\end{align*}
where $\varepsilon'=2\varepsilon + \eta'(\lambda)$ and
\[  C_1= 4 \parens[\Big]{ \parens*{ |I_A|+|I_B| } \, \beta_q^*(\mathcal{G}) + 1 } \parens*{ r_x + \parens*{ \sum_y |G_{xy}| }^2 }\leq 8 \parens[\Big]{ \parens*{ |I_A|+|I_B| } \, \beta_q^*(\mathcal{G}) + 1 } \]
is a constant that depends only on the XOR game.

%\MW{Should we bound $r_x + \parens*{ \sum_y |G_{xy}| }^2 \leq 1 + 1 \leq 2$, to get an overall 8 in front?
%Or at least write $\leq ...$ the simplified expression?}
\end{theorem}

\begin{proof}
Using the triangle inequality, $(x+y)^2\leq 2x^2+2y^2$ and $\|B_y\|\leq 1$, we obtain
    \begin{align*}
\mathop{\mathbb{E}}_{c\leftarrow\mathrm{Enc}(x)}\sum_{\alpha}\|&\left(r_x^2\Id-\left(\sum_y G_{xy}B_y\right)^2\right)\ket{\psi_{c\alpha}}\|^2\\
&\leq 2\mathop{\mathbb{E}}_{c\leftarrow\mathrm{Enc}(x)}\sum_{\alpha}\|r_x^2\Id-(-1)^{\mathrm{Dec}(\alpha)}r_x\sum_y G_{xy}B_y\|^2\\
&\qquad +2\mathop{\mathbb{E}}_{c\leftarrow\mathrm{Enc}(x)}\sum_{\alpha}\|(-1)^{\mathrm{Dec}(\alpha)}r_x\sum_y G_{xy}B_y-\left(\sum_y G_{xy}B_y\right)^2\ket{\psi_{c\alpha}}\|^2\\
&\leq 2r_x\mathop{\mathbb{E}}_{c\leftarrow\mathrm{Enc}(x)}\sum_{\alpha}\|r_x\Id-(-1)^{\mathrm{Dec}(\alpha)}\sum_y G_{xy}B_y\|^2\\
&\qquad +2\mathop{\mathbb{E}}_{c\leftarrow\mathrm{Enc}(x)}\sum_{\alpha}\|\sum_y G_{xy}B_y\|^2\|(-1)^{\mathrm{Dec}(\alpha)}r_x\Id-\sum_y G_{xy}B_y\ket{\psi_{c\alpha}}\|^2\\
&\leq 2\left(r_x+\left(\sum_y |G_{xy}|\right)^2\right)\mathop{\mathbb{E}}_{c\leftarrow\mathrm{Enc}(x)}\sum_{\alpha}\|r_x\Id-(-1)^{\mathrm{Dec}(\alpha)}\sum_y G_{xy}B_y\|^2.
\end{align*}
We deduce
\begin{align*}
\mathop{\mathbb{E}}_{c\leftarrow\mathrm{Enc}(x)}\sum_{\alpha}\|\left(r_x^2\Id-\left(\sum_y G_{xy}B_y\right)^2\right)\ket{\psi_{c\alpha}}\|^2
\leq 2\parens[\Big]{ (|I_A|+|I_B|)\beta_q^*(\mathcal{G})+1 } 2\left(r_x+\left(\sum_y |G_{xy}|\right)^2\right)\varepsilon'
\end{align*}
from \cref{lem:XORselftest}.
\end{proof}

The previous theorem recovers the self-test for anticommuting operators in the CHSH game: Given the CHSH game with a uniform input distribution, we recall
\begin{align*}
    |I_A|=|I_B|=2,\, G_{xy}=\frac{1}{4}(-1)^{xy},\, r_x=\frac{\sqrt{2}}{4} \,\text{ and }\, \beta_q^*(CHSH)=\frac{\sqrt{2}}{2}.
\end{align*}
Then, \cref{thm:XORselftest} yields
\begin{align*}
\mathop{\mathbb{E}}_{c\leftarrow\mathrm{Enc}(x)}\sum_{\alpha}\|\left(B_0B_1+B_1B_0\right)\ket{\psi_{c\alpha}}\|^2\leq 16(2\sqrt{2}+1)(\sqrt{2}+1)(2\varepsilon+\sqrt{|\eta'(\lambda)|})
\end{align*}
for every $\varepsilon$-optimal quantum strategy $S=(\ket{\psi}, \{A_{c\alpha}\}, \{B_{yb}\})$ for the compiled CHSH game.

%-----------------------------------------------------------------------------
\section{Parallel repetition of XOR games}\label{sec:parallel}
%-----------------------------------------------------------------------------
We show an upper bound on the quantum value of the compiled parallel repetition of XOR games.
Throughout this section, we will treat the number of parallel repetitions~$n$ as fixed (in other words, as a constant in the security parameter).
However, we remark that a similar analysis can be used to establish the soundness of any polynomial number~$n$ of parallel repetitions, at the cost of a sub-exponential security guarantee for the QHE.
A similar tradeoff was also shown in~\cite{klvy}.

\begin{definition}
The \emph{parallel repetition}~$\wedge_{i=1}^n \mathcal{G}_i$ of $n$~nonlocal games $\mathcal{G}_i=(I_A^i, I_B^i, O_A^i, O_B^i, \pi_i, V_i)$ is defined as the nonlocal game with input sets $I_A=\bigtimes_{i=1}^nI_A^i$ and $I_B=\bigtimes_{i=1}^nI_B^i$, output sets $O_A=\bigtimes_{i=1}^nO_A^i$ and $O_B=\bigtimes_{i=1}^nO_B^i$, probability distribution $\pi=\otimes_{i=1}^n \pi_i$, and verification function $V(a,b,x,y)=\prod_{i=1}^n V_i(a_i,b_i,x_i,y_i)$.
\end{definition}

While the conjunction of XOR games is not an XOR game, we can make the following definition.
We use the notation $[n]=\{1,\dots,n\}$.

\begin{definition}
Let $\mathcal{G}_i$ be XOR games with cost matrices $G_i$ for~$i\in[n]$, and let~$M\subseteq [n]$.
Then the \emph{sum (modulo two)} $\oplus_{i\in M}\mathcal{G}_i$ is defined as the XOR game with the cost matrix $G=\otimes_{i\in M}G_i$, i.e.\ $G_{xy}=\prod_{i\in M}(G_i)_{x_iy_i}$ for $x=(x_i)_{i\in M}$ and $y=(y_i)_{i\in M}$.
\end{definition}

We first discuss the results of~\cite{Cleve08} for nonlocal games.
Clearly, any quantum strategy~$S$ for~$\wedge_{i=1}^n \mathcal{G}_i$ determines one for~$\oplus_{i\in M}\mathcal{G}_i$, for any subset~$M\subseteq[n]$, as follows:
Alice and Bob first proceed as in strategy~$S$ to obtain~$a,b\in\{0,1\}^n$. They then output~$\oplus_{i=1}^n a_i$ and $\oplus_{i=1}^n b_i$, respectively.
Then one has the following relation between the quantum value of the parallel repetition and the bias of the sum games:

\begin{lemma}[{\cite[Lemma 8]{Cleve08}}]\label{lem:parallel_to_sum}
Let $\mathcal{G}_i$ be XOR games for~$i\in[n]$.
Then it holds for any strategy~$S$ for~$\wedge_{i=1}^n \mathcal{G}_i$ that
\begin{align*}
    \omega_q(S, \wedge_{i=1}^n \mathcal{G}_i)=\frac{1}{2^n}\sum_{M\subseteq [n]} \beta_q(S_M, \oplus_{i\in M}\mathcal{G}_i).
\end{align*}
\end{lemma}

It follows from this and the nontrivial result that the bias of XOR games is multiplicative with respect to taking sums~\cite[Theorem 1]{Cleve08} that one has the upper bound
\begin{align}\label{eq:nonlocal p rep upper}
    \omega_q^*(\wedge_{i=1}^n \mathcal{G}_i)
\leq \frac{1}{2^n}\sum_{M\subseteq [n]} \beta_q^*(\oplus_{i\in M}\mathcal{G}_i)
= \prod_{i=1}^n \omega_q^*(\mathcal G_i),
\end{align}
The converse inequality holds for any nonlocal game, since the players can always play optimal strategies for the individual games~$G_i$ in parallel.
This establishes the parallel repetition theorem for XOR games~\cite[Theorem 2]{Cleve08}:
\begin{align*}
    \omega_q^*(\wedge_{i=1}^n \mathcal{G}_i) = \prod_{i=1}^n \omega_q^*(\mathcal G_i).
\end{align*}

We now use a similar argument to obtain a repetition theorem in the compiled setting.
As above, we first argue that any quantum strategy~$S$ for the compiled parallel repetition~$(\wedge_{i=1}^n \mathcal{G}_i)_\mathrm{comp}$ determines a quantum strategy~$S_M$ for the compile sum games~$(\oplus_{i\in M}\mathcal{G}_i)_\mathrm{comp}$, for any subset~$M\subseteq[n]$.
To this end, we modify step~2 such that if~$S$ were to output~$\alpha$, $S_M$ instead applies $\mathrm{Eval}$ to~$\alpha$ to compute~$\{0,1\}^n \ni a \mapsto \oplus_{i\in M} a_i$ on the encrypted data and returns the new ciphertext~$\alpha'$.%
\footnote{Here we assume that the QHE scheme allows for more than one homomorphic evaluation, i.e., after given an evaluated ciphertext one can keep evaluating homomorphically any circuit of one's choice. This property is satisfied by all known QHE schemes. We omit a formal definition of this property and we refer the reader to \cite{multi-hop} for details.}
% \MW{(Hey Giulio: For \cref{lemma:compiledprobrepetition} below to hold we need that the FHE scheme allows us to apply more than one $\mathrm{Eval}$ query. This doesn't strictly speaking follow from our definition. Should we add a parenthetical remark that we assume this? I think it's fine to assume it, on the other hand it's also not strictly needed -- if we assume that $S$ calls $\mathrm{Eval}$ once then we could probably call it with an oracle $\mathrm{Eval'}$ that on circuit~$C$ behaves like $\mathrm{Eval}$ on circuit $C' = \oplus \circ \mathrm{Eval}$, or something like that, but I'd like to avoid making things overly complicated for the purposes of this note.)}
We also modify step~4 such that if~$S$ were to output~$b\in\{0,1\}^n$, $S_M$ instead outputs~$\oplus_{i \in M} b_i$.
Then the following lemma can be proved similarly as~\cite[Lemma 8]{Cleve08} above:

\begin{lemma}\label{lemma:compiledprobrepetition}
Let $\mathcal{G}_i$ be XOR games for~$i\in[n]$.
Then it holds for any strategy~$S$ for~$(\wedge_{i=1}^n \mathcal{G}_i)_\mathrm{comp}$ that
\begin{align*}
    \omega_q(S, (\wedge_{i=1}^n \mathcal{G}_i)_\mathrm{comp})=\frac{1}{2^n}\sum_{M\subseteq [n]} \beta_q(S_M, (\oplus_{i\in M}\mathcal{G}_i)_\mathrm{comp}).
\end{align*}
\end{lemma}

By combining the preceding with \cref{thm:XORvalue}, we can deduce the following parallel repetition for compiled XOR games.

\begin{theorem}\label{thm:XORparallel}
Let $\mathcal{G}_i$ be XOR games for~$i\in[n]$.
Then for any strategy~$S$ for~$(\wedge_{i=1}^n \mathcal{G}_i)_\mathrm{comp}$ there exists a negligible function~$\eta(\lambda)$ such that the following holds:
\begin{align*}
    \omega_q^*(S, (\wedge_{i=1}^n \mathcal{G}_i)_\mathrm{comp})
% \leq \omega_q^*(\wedge_{i=1}^n \mathcal{G}_i) + \eta(\lambda)
\leq \parens*{ \prod_{i=1}^n \omega_q^*(\mathcal{G}_i) } + \eta(\lambda).
\end{align*}
Conversely, there exists a quantum strategy~$S$ and a negligible function~$\eta(\lambda)$ such that
\begin{align*}
    \omega_q^*(S, (\wedge_{i=1}^n \mathcal{G}_i)_\mathrm{comp})
\geq \parens*{ \prod_{i=1}^n \omega_q^*(\mathcal{G}_i) } - \eta(\lambda).
\end{align*}
\end{theorem}

\begin{proof}
By \cref{lemma:compiledprobrepetition}, we have
\begin{align*}
    \omega_q(S, (\wedge_{i=1}^n \mathcal{G}_i)_\mathrm{comp}) = \frac{1}{2^n}\sum_{M\subseteq [n]} \beta_q(S_M, (\oplus_{i\in M}\mathcal{G}_i)_\mathrm{comp}).
\end{align*}
Since each~$\oplus_{i\in M}\mathcal{G}_i$ is an XOR game, and remembering the relation between the bias and the winning probability, we can apply \cref{thm:XORvalue} and deduce from it that there exist negligible functions~$\eta_M(\lambda)$ such that
\begin{align*}
    \beta_q(S_M, (\oplus_{i\in M}\mathcal{G}_i)_\mathrm{comp})
\leq \beta^*_q(\oplus_{i\in M}\mathcal{G}_i) + \eta_M(\lambda),
\end{align*}
Using this and \cref{eq:nonlocal p rep upper}, we obtain that
\begin{align*}
    \omega_q(S, (\wedge_{i=1}^n \mathcal{G}_i)_\mathrm{comp})
\leq \parens*{ \frac{1}{2^n}\sum_{M\subseteq [n]} \beta^*_q(\oplus_{i\in M}\mathcal{G}_i) } + \eta(\lambda)
= \parens*{ \prod_{i=1}^n \omega_q^*(\mathcal G_i) } + \eta(\lambda),
\end{align*}
where $\eta := \frac1{2^n} \sum_{M \subseteq [n]} \eta_M$ is a negligible function.
This establishes the first claim.

The second follows at once from \cref{thm:KLVY} and the inequality $\omega_q^*(S, \wedge_{i=1}^n \mathcal{G}_i) \geq \prod_{i=1}^n \omega_q^*(S, \mathcal G_i)$ which as mentioned above holds for any nonlocal game.
\end{proof}

%=============================================================================
\section{Magic square game}\label{sec:magic}
%=============================================================================

Consider the following set of equations
\begin{align}\label{eq:Magicsquare}
    &x_1+x_2+x_3=0 \quad(r_1), &&x_1+x_4+x_7=0\quad(c_1),\nonumber\\
    &x_4+x_5+x_6=0\quad(r_2), &&x_2+x_5+x_8=0\quad(c_2),\\
    &x_7+x_8+x_9=0\quad(r_3), &&x_3+x_6+x_9=1\quad(c_3)\nonumber.
\end{align}
In the nonlocal \emph{magic square game}~$\mathcal{G}_\mathrm{MS}$, the referee sends the index~$i\in [6]$ of one of the six equations to Alice, and the index~$j \in S_i$ of one of three variables in Alice's equation to Bob.
Here, $S_i \subseteq [9]$ denotes the index of the variables that appear in the $i$th equation.
That is, $S_i = \{s_1,s_2,s_3\}$ if the $i$th equation contains variables~$x_{s_1},x_{s_2},x_{s_3}$.
Alice answers with $\{0,1\}$-assignment $a=(a_{s_1},a_{s_2},a_{s_3})$ and Bob answers with an $\{0,1\}$-assignment $b$ to his variable~$x_j$.
The players win the game, Alice's assignments satisfy her equation and if two players' assignments coincide in the common variable~$x_j$.

The magic square game has a perfect quantum strategy, that is, $\omega_q^*(\mathcal{G}_\mathrm{MS})=1$.
Hence the same is true for the compiled game by \cref{thm:KLVY}.

\begin{corollary}
  There exists a quantum strategy~$S$ for the compiled game~$(\mathcal{G}_\mathrm{MS})_\mathrm{comp}$ and a negligible function~$\eta(\lambda)$ such that
  \[ \omega_q(S, (\mathcal{G}_\mathrm{MS})_\mathrm{comp}) \geq 1 - \eta(\lambda). \]
\end{corollary}

In the remainder of this section, we want to prove that in any almost perfect quantum strategy for the compiled magic square game, we can find almost anticommuting Bob operators.
We first make a definition.

\begin{definition}
We call a quantum strategy~$S$ for~$(\mathcal{G}_\mathrm{MS})_\mathrm{comp}$ \emph{$\eps$-perfect} if (for sufficiently large~$\lambda$)
\begin{align*}
    \omega_q(S, (\mathcal{G}_\mathrm{MS})_\mathrm{comp}) \geq 1 - \eps.
\end{align*}
\end{definition}

Note that for any $\eps$-perfect quantum strategy $S=(\ket{\psi}, \{A_{c\alpha}\}, \{B_{jb}\})$ of the compiled magic square game, it holds
\begin{align*}
    \sum_{\substack{a,b,i,j;\\V(a,b,i,j)=0}}p(a,b|i,j)\leq 18\varepsilon,
\end{align*}
where~$p(a,b|i,j)$ denotes the probability that the prover returns (an encryption of)~$a$ and~$b$ when given as questions (the encryption of)~$i$ and~$j$.
That is, if the $i$th equation be of the form $x_{s_1}+x_{s_2}+x_{s_3}=d_i$. Then, we have
\begin{align}\label{eq:almostperfect}
\sum_i\mathop{\mathbb{E}}_{c\leftarrow\mathrm{Enc}(i)}\Bigg(\sum_{j\in S_i}\sum_{\substack{a;\\a_{s_1}+a_{s_2}+a_{s_3}= d_i}}&\sum_{\alpha; \mathrm{Dec}(\alpha)=a}\|B_{j(a_j\op1)}\ket{\psi_{c\alpha}}\|^2\nonumber\\
&+\sum_{\substack{a;\\a_{s_1}+a_{s_2}+a_{s_3}\neq d_i}}\sum_{\alpha; \mathrm{Dec}(\alpha)=a}\|\ket{\psi_{c\alpha}}\|^2\Bigg)\leq 18\varepsilon,
\end{align}
since $V(a,a_j\op 1,i,j)=0$ and $V(a,b,i,j)=0$ for all $b$ if $a_{s_1}+a_{s_2}+a_{s_3}\neq d_i$.

We now prove a number of technical lemmas.
Define $B_j:=B_{j0}-B_{j1}$.

\begin{lemma}\label{lemma:MS1}
Let $S=(\ket{\psi}, \{A_{c\alpha}\}, \{B_{jb}\})$ be a quantum strategy for the compiled magic square game. Then
\begin{align*}
\|(B_{j}-(-1)^{\mathrm{Dec}(\alpha)_j})\ket{\psi_{c\alpha}}\|^2=4\|B_{j(\mathrm{Dec}(\alpha)_j\op1)}\ket{\psi_{c\alpha}}\|^2.
\end{align*}
\end{lemma}

\begin{proof}
Using $\ket{\psi_{c\alpha}}=(B_{j0}+B_{j1})\ket{\psi_{c\alpha}}$, we have
\begin{align*}
 \|(B_{j}-(-1)^{\mathrm{Dec}(\alpha)_j})\ket{\psi_{c\alpha}}\|^2&=\|(1-(-1)^{\mathrm{Dec}(\alpha)_j})B_{j0}+(-1-(-1)^{\mathrm{Dec}(\alpha)_j})B_{j1}\ket{\psi_{c\alpha}}\|^2\\
 &=4\|B_{j(\mathrm{Dec}(\alpha)_j\op1)}\ket{\psi_{c\alpha}}\|^2.
\end{align*}
\end{proof}

\begin{lemma}\label{lemma:MS2}
Let $x_{s_1}+x_{s_2}+x_{s_3}=d_i$ be the $i$th equation in the magic square game and let~$S=(\ket{\psi}, \{A_{c\alpha}\}, \{B_{jb}\})$ be an $\varepsilon$-perfect quantum strategy for the compiled magic square game.
Then,
\begin{align*}
  \mathop{\mathbb{E}}_{c\leftarrow\mathrm{Enc}(i)}\sum_{\alpha}\|(B_{s_1}B_{s_2}-(-1)^{d_i}B_{s_3})\ket{\psi_{c\alpha}}\|^2\leq 216\varepsilon.
\end{align*}
\end{lemma}

\begin{proof}
For $\alpha$ with $\mathrm{Dec}(\alpha)=a$, $a_{s_1}+a_{s_2}+a_{s_3}= d_i$, we compute\begin{align*}
\|(B_{s_1}B_{s_2}-(-1)^{d_i}B_{s_3})\ket{\psi_{c\alpha}}\|^2&=\|(B_{s_1}B_{s_2}+\left(-(-1)^{a_{s_2}}B_{s_1}+(-1)^{a_{s_2}}B_{s_1}\right)\\
&\qquad+\left(-(-1)^{d_i}(-1)^{a_{s_3}}+(-1)^{d_i}(-1)^{a_{s_3}}\right) -(-1)^{d_i}B_{s_3})\ket{\psi_{c\alpha}}\|^2\\
&\leq (\|(B_{s_1}B_{s_2}-(-1)^{a_{s_2}}B_{s_1})\ket{\psi_{c\alpha}}\|\\
&\qquad +\|((-1)^{a_{s_2}}B_{s_1}-(-1)^{d_i}(-1)^{a_{s_3}})\ket{\psi_{c\alpha}}\|\\
&\qquad+\|((-1)^{d_i}(-1)^{a_{s_3}}-(-1)^{d_i}B_{s_3})\ket{\psi_{c\alpha}}\|)^2\\
&\leq (\|B_{s_1}(B_{s_2}-(-1)^{a_{s_2}})\ket{\psi_{c\alpha}}\|+\|(B_{s_1}-(-1)^{a_{s_1}})\ket{\psi_{c\alpha}}\|\\
&\qquad+\|((-1)^{a_{s_3}}-B_{s_3})\ket{\psi_{c\alpha}}\|)^2\\
&\leq 3\|B_{s_1}\|^2\|(B_{s_2}-(-1)^{a_{s_2}})\ket{\psi_{c\alpha}}\|^2+3\|(B_{s_1}-(-1)^{a_{s_1}})\ket{\psi_{c\alpha}}\|^2\\
&\qquad+3\|((-1)^{a_{s_3}}-B_{s_3})\ket{\psi_{c\alpha}}\|^2\\
&=12\sum_{j\in S_i}\|B_{j(\mathrm{Dec}(\alpha)_j\op1)}\ket{\psi_{c\alpha}}\|^2,
\end{align*}
by using $(-1)^{d_i}(-1)^{a_{s_3}}=(-1)^{a_{s_1}}(-1)^{a_{s_2}}$, $(x+y+z)^2\leq 3x^2+3y^2+3z^2$ (by the Cauchy-Schwarz inequality) and \cref{lemma:MS1}. For $\alpha$ with $\mathrm{Dec}(\alpha)=a$, $a_{s_1}+a_{s_2}+a_{s_3}\neq d_i$, it holds
\begin{align*}
    \|(B_{s_1}B_{s_2}-(-1)^{d_i}B_{s_3})\ket{\psi_{c\alpha}}\|^2\leq 4\|\ket{\psi_{c\alpha}}\|^2,
\end{align*}
since $\|B_{s_1}B_{s_2}-(-1)^{d_i}B_{s_3}\|\leq \|B_{s_1}\|\|B_{s_2}\|+\|B_{s_3}\|\leq 2$. Therefore, we obtain
\begin{align}\label{eq:inproofMS2}
    \mathop{\mathbb{E}}_{c\leftarrow\mathrm{Enc}(i)}&\sum_{\alpha}\|(B_{s_1}B_{s_2}-(-1)^{d_i}B_{s_3})\ket{\psi_{c\alpha}}\|^2\nonumber\\
   &\leq \mathop{\mathbb{E}}_{c\leftarrow\mathrm{Enc}(i)}\Bigg(12\sum_{j\in S_i}\sum_{\substack{a;\\a_{s_1}+a_{s_2}+a_{s_3}= d_i}}\sum_{\alpha; \mathrm{Dec}(\alpha)=a}\|B_{j(a_j+1)}\ket{\psi_{c\alpha}}\|^2\nonumber\\
  &\qquad \qquad \qquad \qquad+4\sum_{\substack{a;\\a_{s_1}+a_{s_2}+a_{s_3}\neq d_i}}\sum_{\alpha; \mathrm{Dec}(\alpha)=a}\|\ket{\psi_{c\alpha}}\|^2\Bigg)\\
 & \leq 216 \varepsilon\nonumber
\end{align}
by \cref{eq:almostperfect}.
\end{proof}

\begin{lemma}\label{lemma:MS3}
Let $x_{s_1}+x_{s_2}+x_{s_3}=d_i$ be the $i$th equation in the magic square game and let~$S=(\ket{\psi}, \{A_{c\alpha}\}, \{B_{jb}\})$ be an $\varepsilon$-perfect strategy for the compiled magic square game. Let $t\in S_k$.
Then, there exists a negligible function~$\eta(\lambda)$ such that
\begin{align*}
  \mathop{\mathbb{E}}_{c\leftarrow\mathrm{Enc}(k)}\sum_{\alpha}\|(B_{s_1}B_{s_2}-(-1)^{d_i}B_{s_3})B_t\ket{\psi_{c\alpha}}\|^2\leq 32\varepsilon+2\eta(\lambda).
\end{align*}
\end{lemma}

\begin{proof}
Using $(x+y)^2\leq 2x^2+2y^2$, it holds
\begin{align*}
    \mathop{\mathbb{E}}_{c\leftarrow\mathrm{Enc}(k)}\sum_{\alpha}\|(B_{s_1}B_{s_2}&-(-1)^{d_i}B_{s_3})B_t\ket{\psi_{c\alpha}}\|^2\\
    &\leq 2\mathop{\mathbb{E}}_{c\leftarrow\mathrm{Enc}(k)}\sum_{\alpha}\|(B_{s_1}B_{s_2}-(-1)^{d_i}B_{s_3})\|^2\|(B_t-(-1)^{\mathrm{Dec}(\alpha)_t})\ket{\psi_{c\alpha}}\|^2\\
    &\qquad+2\mathop{\mathbb{E}}_{c\leftarrow\mathrm{Enc}(k)}\sum_{\alpha}\|(B_{s_1}B_{s_2}-(-1)^{d_i}B_{s_3})\ket{\psi_{c\alpha}}\|^2.
\end{align*}
By \cref{lemma:MS1,eq:almostperfect}, we have \allowdisplaybreaks
\begin{align*}
    \mathop{\mathbb{E}}_{c\leftarrow\mathrm{Enc}(k)}\sum_{\alpha}\|(B_{s_1}B_{s_2}&-(-1)^{d_i}B_{s_3})\|^2\|(B_t-(-1)^{\mathrm{Dec}(\alpha)_t})\ket{\psi_{c\alpha}}\|^2\\
    &\leq  4\mathop{\mathbb{E}}_{c\leftarrow\mathrm{Enc}(k)}\sum_{\alpha}\|(B_{t}-(-1)^{\mathrm{Dec}(\alpha)_t})\ket{\psi_{c\alpha}}\|^2\\
    &=16\mathop{\mathbb{E}}_{c\leftarrow\mathrm{Enc}(k)}\sum_{\alpha}\|B_{t(\mathrm{Dec}(\alpha)_t+1)}\ket{\psi_{c\alpha}}\|^2\\
  &=16\mathop{\mathbb{E}}_{c\leftarrow\mathrm{Enc}(k)}\sum_{\substack{a;\\a_{s_1}+a_{s_2}+a_{s_3}= d_i}}\sum_{\alpha; \mathrm{Dec}(\alpha)=a}\|B_{t(\mathrm{Dec}(\alpha)_t+1)}\ket{\psi_{c\alpha}}\|^2\\
  &\qquad + 16\mathop{\mathbb{E}}_{c\leftarrow\mathrm{Enc}(k)}\sum_{\substack{a;\\a_{s_1}+a_{s_2}+a_{s_3}\neq d_i}}\sum_{\alpha; \mathrm{Dec}(\alpha)=a}\|B_{t(\mathrm{Dec}(\alpha)_t+1)}\ket{\psi_{c\alpha}}\|^2\\
  &\leq 16\mathop{\mathbb{E}}_{c\leftarrow\mathrm{Enc}(k)}\sum_{\substack{a;\\a_{s_1}+a_{s_2}+a_{s_3}= d_i}}\sum_{\alpha; \mathrm{Dec}(\alpha)=a}\|B_{t(\mathrm{Dec}(\alpha)_t+1)}\ket{\psi_{c\alpha}}\|^2\\
  &\qquad + 16\mathop{\mathbb{E}}_{c\leftarrow\mathrm{Enc}(k)}\sum_{\substack{a;\\a_{s_1}+a_{s_2}+a_{s_3}\neq d_i}}\sum_{\alpha; \mathrm{Dec}(\alpha)=a}\|\ket{\psi_{c\alpha}}\|^2\\
  &\leq 288\varepsilon.
\end{align*}
Now, note that $(B_{s_1}B_{s_2}-(-1)^{d_i}B_{s_3})^2$ has a QPT-implementable block encoding with scale factor $\Theta(1)$ by \cref{lem:block lincom,lem:block prod}. Furthermore $\|(B_{s_1}B_{s_2}-(-1)^{d_i}B_{s_3})^2\|\leq 4$. Using \cref{lemma:Strongnonsignalling} and then \cref{eq:inproofMS2} in the proof of \cref{lemma:MS2}, we get
\begin{align*}
\mathop{\mathbb{E}}_{c\leftarrow\mathrm{Enc}(k)}\sum_{\alpha}\|(B_{s_1}B_{s_2}&-(-1)^{d_i}B_{s_3})\ket{\psi_{c\alpha}}\|^2\\
&\leq \mathop{\mathbb{E}}_{c\leftarrow\mathrm{Enc}(i)}\sum_{\alpha}\|(B_{s_1}B_{s_2}-(-1)^{d_i}B_{s_3})\ket{\psi_{c\alpha}}\|^2+\eta(\lambda)\\
&\leq \mathop{\mathbb{E}}_{c\leftarrow\mathrm{Enc}(i)}\Bigg(12\sum_{j\in S_i}\sum_{\substack{a;\\a_{s_1}+a_{s_2}+a_{s_3}= d_i}}\sum_{\alpha; \mathrm{Dec}(\alpha)=a}\|B_{j(a_j+1)}\ket{\psi_{c\alpha}}\|^2\nonumber\\
  &\qquad \qquad \qquad \qquad+4\sum_{\substack{a;\\a_{s_1}+a_{s_2}+a_{s_3}\neq d_i}}\sum_{\alpha; \mathrm{Dec}(\alpha)=a}\|\ket{\psi_{c\alpha}}\|^2\Bigg) +\eta(\lambda).
\end{align*}
Summarizing, we obtain
\begin{align*}
  \mathop{\mathbb{E}}_{c\leftarrow\mathrm{Enc}(k)}\sum_{\alpha}\|(B_{s_1}B_{s_2}-(-1)^{d_i}B_{s_3})B_t\ket{\psi_{c\alpha}}\|^2\leq 576\varepsilon+2\eta(\lambda)
\end{align*}
from \cref{eq:almostperfect}.
\end{proof}

\begin{lemma}\label{lemma:MS4}
Let $x_{s_1}+x_{s_2}+x_{s_3}=d_i$ and $x_{t_1}+x_{t_2}+x_{t_3}=d_j$ be the $i$th equation and $j$th equation in the magic square game, respectively. Furthermore, let $S=(\ket{\psi}, \{A_{c\alpha}\}, \{B_{jb}\})$ be an $\varepsilon$-perfect strategy for the compiled magic square game.
Then, there exists a negligible function~$\eta(\lambda)$ such that
\begin{align*}
  \mathop{\mathbb{E}}_{c\leftarrow\mathrm{Enc}(j)}\sum_{\alpha}\|(B_{s_1}B_{t_1}-(-1)^{d_i+d_j}B_{s_2}B_{s_3}B_{t_2}B_{t_3})\ket{\psi_{c\alpha}}\|^2\leq 1584\varepsilon+4\eta(\lambda).
\end{align*}
\end{lemma}

\begin{proof}
We have
\begin{align*}
 \mathop{\mathbb{E}}_{c\leftarrow\mathrm{Enc}(j)}\sum_{\alpha}\|(B_{s_1}B_{t_1}&-(-1)^{d_i+d_j}B_{s_2}B_{s_3}B_{t_2}B_{t_3})\ket{\psi_{c\alpha}}\|^2\\
 &\leq 2\mathop{\mathbb{E}}_{c\leftarrow\mathrm{Enc}(j)}\sum_{\alpha}\|(B_{s_1}B_{t_1}-(-1)^{d_i}B_{s_2}B_{s_3}B_{t_1})\ket{\psi_{c\alpha}}\|^2\\
 &\qquad +2\mathop{\mathbb{E}}_{c\leftarrow\mathrm{Enc}(j)}\sum_{\alpha}\|(B_{s_2}B_{s_3}B_{t_1}-(-1)^{d_j}B_{s_2}B_{s_3}B_{t_2}B_{t_3})\ket{\psi_{c\alpha}}\|^2.
\end{align*}
\cref{lemma:MS3} yields
\begin{align*}
    \mathop{\mathbb{E}}_{c\leftarrow\mathrm{Enc}(j)}\sum_{\alpha}\|(B_{s_1}B_{t_1}-(-1)^{d_i}B_{s_2}B_{s_3}B_{t_1})\ket{\psi_{c\alpha}}\|^2\leq576\varepsilon+2\eta(\lambda).
\end{align*}
Furthermore, using \cref{lemma:MS2} and $\|B_{s_k}\|\leq 1$, we get
\begin{align*}
    \mathop{\mathbb{E}}_{c\leftarrow\mathrm{Enc}(j)}\sum_{\alpha}\|(B_{s_2}B_{s_3}B_{t_1}&-(-1)^{d_j}B_{s_2}B_{s_3}B_{t_2}B_{t_3})\ket{\psi_{c\alpha}}\|^2\\
    &\leq \mathop{\mathbb{E}}_{c\leftarrow\mathrm{Enc}(j)}\sum_{\alpha}\|B_{s_2}\|^2\|B_{s_3}\|^2\|(B_{t_1}-(-1)^{d_j}B_{t_2}B_{t_3})\ket{\psi_{c\alpha}}\|^2\\
    &\leq \mathop{\mathbb{E}}_{c\leftarrow\mathrm{Enc}(j)}\sum_{\alpha}\|(-(-1)^{d_j}B_{t_1}+B_{t_2}B_{t_3})\ket{\psi_{c\alpha}}\|^2\\
    &\leq 216 \varepsilon.
\end{align*}
This yields the assertion.
\end{proof}

We are now ready to prove the main theorem of this section. Note that the operators $B_2$ and $B_4$ are chosen purely for convenience and can be changed to other Bob observables, due to symmetry of the magic square game.

\begin{theorem}\label{thm:magic}
Let $S=(\ket{\psi}, \{A_{c\alpha}\}, \{B_{jb}\})$ be an $\varepsilon$-perfect strategy of the compiled magic square game.
Then, there exists a negligible function~$\eta(\lambda)$ such that
\begin{align*}
  \mathop{\mathbb{E}}_{c\leftarrow\mathrm{Enc}(i)}\sum_{\alpha}\|(B_2B_4+B_4B_2)\ket{\psi_{c\alpha}}\|^2\leq  17280\varepsilon+52\eta(\lambda)
\end{align*}
for all $i\in [6]$.
\end{theorem}

\begin{proof}
Let $r_i,c_i$ be as in \cref{eq:Magicsquare}. We start by computing
\begin{align}
\mathop{\mathbb{E}}_{c\leftarrow\mathrm{Enc}(r_2)}\sum_{\alpha}\|(B_2B_4+&B_1B_9B_5)\ket{\psi_{c\alpha}}\|^2\nonumber\\
&\leq 2\mathop{\mathbb{E}}_{c\leftarrow\mathrm{Enc}(r_2)}\sum_{\alpha}\|(B_2B_4-B_1B_3B_6B_5)\ket{\psi_{c\alpha}}\|^2\nonumber\\
&\qquad +2 \mathop{\mathbb{E}}_{c\leftarrow\mathrm{Enc}(r_2)}\sum_{\alpha}\|B_1\|^2\|(B_3B_6+B_9)B_5\ket{\psi_{c\alpha}}\|^2\nonumber\\
&\leq 4320\varepsilon+12\eta(\lambda)\label{eq:anticomm1}
\end{align}
by using $(x+y)^2\leq 2x^2+2y^2$, \cref{lemma:MS4,lemma:MS3}. Similarly, we obtain
\begin{align}
\mathop{\mathbb{E}}_{c\leftarrow\mathrm{Enc}(c_2)}\sum_{\alpha}\|(B_4B_2-&B_1B_9B_5)\ket{\psi_{c\alpha}}\|^2\nonumber\\
&\leq 2\mathop{\mathbb{E}}_{c\leftarrow\mathrm{Enc}(c_2)}\sum_{\alpha}\|(B_4B_2-B_1B_7B_8B_5)\ket{\psi_{c\alpha}}\|^2\nonumber\\
&\qquad +2 \mathop{\mathbb{E}}_{c\leftarrow\mathrm{Enc}(c_2)}\sum_{\alpha}\|B_1\|^2\|(B_7B_8-B_9)B_5\ket{\psi_{c\alpha}}\|^2\nonumber\\
&\leq 4320\varepsilon+12\eta(\lambda).\label{eq:anticomm2}
\end{align}
Note that $(B_4B_2-B_1B_9B_5)^2$ has a QPT-implementable block encoding with scale factor $\Theta(1)$ by \cref{lem:block lincom,lem:block prod} and $\|(B_4B_2-B_1B_9B_5)^2\|\leq 4$. Thus, again using $(x+y)^2\leq 2x^2+2y^2$ and \cref{lemma:Strongnonsignalling}, we get
\begin{align*}
    \mathop{\mathbb{E}}_{c\leftarrow\mathrm{Enc}(i)}\sum_{\alpha}\|(B_2B_4+&B_4B_2)\ket{\psi_{c\alpha}}\|^2\\
    &\leq 2\mathop{\mathbb{E}}_{c\leftarrow\mathrm{Enc}(r_2)}\sum_{\alpha}\|(B_2B_4+B_1B_9B_5)\ket{\psi_{c\alpha}}\|^2\\
    &\qquad +2\mathop{\mathbb{E}}_{c\leftarrow\mathrm{Enc}(c_2)}\sum_{\alpha}\|(B_4B_2-B_1B_9B_5)\ket{\psi_{c\alpha}}\|^2+4\eta(\lambda).
\end{align*}
Then, by \cref{eq:anticomm1,eq:anticomm2}, we obtain
\begin{align*}
\mathop{\mathbb{E}}_{c\leftarrow\mathrm{Enc}(i)}\sum_{\alpha}\|(B_2B_4+&B_4B_2)\ket{\psi_{c\alpha}}\|^2\leq 17280\varepsilon+ 52\eta(\lambda).
\end{align*}
\end{proof}

\section*{Acknowledgments}
GM, SS and MW acknowledge support by the Deutsche Forschungsgemeinschaft (DFG, German Research Foundation) under Germany's Excellence Strategy - EXC\ 2092\ CASA - 390781972.
GM also acknowledges support by the European Research Council through an ERC Starting Grant (Grant agreement No.~101077455, ObfusQation).
MW also acknowledges support by the European Research Council through an ERC Starting Grant (Grant agreement No.~101040907, SYMOPTIC), by the NWO through grant OCENW.KLEIN.267, by the German Federal Ministry of Research, Technology and Space (QuBRA, 13N16135; QuSol, 13N17173), and by the Deutsche Forschungsgemeinschaft (DFG, German Research Foundation, 556164098).
MW acknowledges the Simons Institute for the Theory of Computing at UC Berkeley for its hospitality and support while this paper was completed. AM and CP are supported by NSERC Alliance Consortia Quantum grants, reference number number: ALLRP 578455 - 22. AN and CP acknowledge SLMATH for hospitality and financial support during the MIP* = RE hot topics workshop. AN and TZ acknowledge the Simons Institute at UC Berkeley, on whose whiteboard a crucial calculation was performed.

\bibliographystyle{plainurl}
\bibliography{bibliography}

\appendix

\section{Alternate proof of the upper bound}\label{appendix}

% %-----------------------------------------------------------------------------
% \subsection{The }\label{sec:main}
% %-----------------------------------------------------------------------------
Here we give an alternative proof of the upper bound on the quantum value of compiled XOR game (\cref{thm:XORvalue}). This alternate approach starts with the following lemma.

\begin{lemma}\label{lemma:bias}
Let $\mathcal G$ be an XOR game, with cost matrix~$G=(G_{xy})_{x \in I_A, y \in I_B}$, and let~$V \in \R^{I_B \times I_B}$ be a positive semidefinite matrix such that~$V_{yy} = 1$ for all~$y \in I_B$.
Then there exist unit vectors~$u_x,v_y \in \R^d$ for some~$d\in\N$ such that, for all~$x \in I_A$,
\begin{align}\label{eq:vector identity}
    \sqrt{\sum_{s,y \in I_B} G_{xs} G_{xy} V_{sy}}
= \sum_{y \in I_B} G_{xy} \braket{u_x}{v_y}.
\end{align}
\end{lemma}
% \MW{Alternativ rendering of the same proof:}
% \begin{proof}
% Since the matrix~$V\in\R^{I_B \times I_B}$ is positive semi-definite, we can write it as $V = W^T W$ for some matrix $W\in\R^{I_B \times I_B}$.
% Then $u_x := \frac {W G^T \ket x}{\norm{W G^T \ket x}}$ and $v_y := W \ket y$ for $x \in I_A, y \in I_B$ satisfy the desired properties.
% Indeed, the $u_x$ are clearly unit vectors, and since~$\norm{v_y}^2 = V_yy = 1$ the same is true for the~$v_y$.
% Finally,
% \begin{align*}
%   \sum_{y \in I_B} G_{xy} \braket{v_y}{u_x}
% = \sum_{y \in I_B} G_{xy} \bra y W^T \frac {W G^T \ket x}{\norm{W G^T \ket x}}
% = \bra x G W^T \frac {W G^T \ket x}{\norm{W G^T \ket x}}
% = \norm{W G^T \ket x}
% % = \sqrt{\bra x G V G^T \ket x}
% = \sqrt{\sum_{s,y \in I_Y} G_{xs} G_{xy} V_{sy}}
% \end{align*}
% \end{proof}
% \MW{Original proof:}
\begin{proof}
Since the matrix~$V\in \R^{I_B \times I_B}$ is positive semi-definite, there exists~$d\in\N$ and vectors~$v_y\in \R^d$ for~$y\in I_B$ such that $V_{sy}=\braket{v_s}{v_y}$ for all~$s,y \in I_B$.
Each~$v_y$ is a unit vector, since we have~$\norm{v_y}^2 = \braket{v_y}{v_y} = V_{yy} = 1$ by assumption.
Next, for every~$x\in I_A$, define unit vectors
\[ u_x=\frac{1}{\|\sum_{s \in I_B} G_{xs}v_s\|}\sum_{s\in I_B} G_{xs}v_s\in \R^d. \]
Then,
\begin{align*}
    \sum_{y \in I_B} G_{xy} \braket{u_x}{v_y}
= \frac{1}{\|\sum_{s \in I_B} G_{xs}v_s\|} \sum_{s,y\in I_B} G_{xs} G_{xy} V_{sy}
= \sqrt{\sum_{s,y\in I_B} G_{xs} G_{xy} V_{sy}},
\end{align*}
since $\|\sum_{s \in I_B} G_{xs}v_s\|^2 = \sum_{s,y \in I_B} G_{xs} G_{xy} V_{sy}$.
Thus we have established \cref{eq:vector identity}.
\end{proof}

% In the following we write~$B_y:=B_{y0}-B_{y1}$ for the observables corresponding to a PVM~$B_{yb}$.
The next lemma is proven in the same way as~\cite[Lemma 26]{nz23}.

\begin{lemma}\label{lemma:Jensen}
Let $\mathcal{G}_\mathrm{comp}$ be a compiled XOR game and $S=(\ket{\psi}, \{A_{c\alpha}\}, \{B_{yb}\})$ a quantum strategy.
For every~$x\in I_A$
\begin{align*}
\left|\mathop{\mathbb{E}}_{c\leftarrow\mathrm{Enc}(x)}\sum_{\alpha}(-1)^{\mathrm{Dec}(\alpha)}\bra{\psi_{c\alpha}} \left(\sum_yG_{xy}B_y\right)\ket{\psi_{c\alpha}}\right|^2\leq \mathop{\mathbb{E}}_{c\leftarrow\mathrm{Enc}(x)}\sum_{\alpha}\bra{\psi_{c\alpha}} \left(\sum_yG_{xy}B_y\right)^2\ket{\psi_{c\alpha}},
\end{align*}
where~$G = (G_{xy})$ is the cost matrix of the underlying XOR game, $\ket{\psi_{c\alpha}} = A_{c\alpha} \ket\psi$ are the post-measurement states after Step~2 of the protocol, and the~$B_y := B_{y0} - B_{y1}$ are the (binary) observables corresponding to the PVMs~$\{B_{yb}\}_b$.
\end{lemma}

\begin{proof}
Applying Jensen's inequality to the function $g:\R\to\R$, $x\mapsto x^2$, we obtain the three inequalities
\begin{align}
    (\mathbb{E}(X))^2&\leq \mathbb{E}(X^2) &&\text{for real-valued random variables } X,\label{eq:Jensen1}\\
    \left(\sum_k p_k x_k\right)^2&\leq \sum_k p_k x_k^2 &&\text{for } x_k \in \R, p_k\geq 0 \text{ and } \sum_k p_k=1,\label{eq:Jensen2}\\
    (\bra{\psi}O\ket{\psi})^2&\leq \bra{\psi}O^2\ket{\psi} &&\text{for quantum states } \ket{\psi} \text{ and observables }O.\label{eq:Jensen3}
\end{align}
Define $\mathrm{B}_x=\sum_yG_{xy}B_y$. From \cref{eq:Jensen1}, we deduce
\begin{align*}
\left|\mathop{\mathbb{E}}_{c\leftarrow\mathrm{Enc}(x)}\sum_{\alpha}(-1)^{\mathrm{Dec}(\alpha)}\bra{\psi_{c\alpha}} \mathrm{B}_x\ket{\psi_{c\alpha}}\right|^2\leq \mathop{\mathbb{E}}_{c\leftarrow\mathrm{Enc}(x)}\left(\sum_{\alpha}(-1)^{\mathrm{Dec}(\alpha)}\bra{\psi_{c\alpha}}\mathrm{B}_x\ket{\psi_{c\alpha}}\right)^2.
\end{align*}
Note that the post-measurement states $\ket{\psi_{c\alpha}}$ are not normalized and it holds $\sum_\alpha \|\ket{\psi_{c\alpha}}\|^2=1$. Therefore, using \cref{eq:Jensen2}, we get
\begin{align*}
\mathop{\mathbb{E}}_{c\leftarrow\mathrm{Enc}(x)}\left(\sum_{\alpha}(-1)^{\mathrm{Dec}(\alpha)}\bra{\psi_{c\alpha}}\mathrm{B}_x\ket{\psi_{c\alpha}}\right)^2&=\mathop{\mathbb{E}}_{c\leftarrow\mathrm{Enc}(x)}\left(\sum_{\alpha} \|\ket{\psi_{c\alpha}}\|^2(-1)^{\mathrm{Dec}(\alpha)}\frac{\bra{\psi_{c\alpha}}\mathrm{B}_x\ket{\psi_{c\alpha}}}{\|\ket{\psi_{c\alpha}}\|^2}\right)^2\\
&\leq \mathop{\mathbb{E}}_{c\leftarrow\mathrm{Enc}(x)}\sum_{\alpha} \|\ket{\psi_{c\alpha}}\|^2\left(\frac{\bra{\psi_{c\alpha}}\mathrm{B}_x\ket{\psi_{c\alpha}}}{\|\ket{\psi_{c\alpha}}\|^2}\right)^2.
\end{align*}
Since $\frac{\ket{\psi_{c\alpha}}}{\|\ket{\psi_{c\alpha}}\|}$ is a quantum state, we finally infer
\begin{align*}
\mathop{\mathbb{E}}_{c\leftarrow\mathrm{Enc}(x)}\sum_{\alpha} \|\ket{\psi_{c\alpha}}\|^2\left(\frac{\bra{\psi_{c\alpha}}\mathrm{B}_x\ket{\psi_{c\alpha}}}{\|\ket{\psi_{c\alpha}}\|^2}\right)^2&\leq \mathop{\mathbb{E}}_{c\leftarrow\mathrm{Enc}(x)}\sum_{\alpha} \|\ket{\psi_{c\alpha}}\|^2\frac{\bra{\psi_{c\alpha}}\mathrm{B}_x^2\ket{\psi_{c\alpha}}}{\|\ket{\psi_{c\alpha}}\|^2}\\
&=\mathop{\mathbb{E}}_{c\leftarrow\mathrm{Enc}(x)}\sum_{\alpha} \bra{\psi_{c\alpha}}\mathrm{B}_x^2\ket{\psi_{c\alpha}}
\end{align*}
from \cref{eq:Jensen3}.
\end{proof}

We are now ready to give the alternative proof of our main result.

\begin{proof}[Alternate proof of \cref{thm:XORvalue}]
%We only need to prove an upper bound, since we know $\omega_q^*(\mathcal{G}_\mathrm{comp})\geq\omega_q^*(\mathcal{G})+\eta(\lambda)$ from \cref{thm:KLVY}.
Let $S=(\ket{\psi}, \{A_{c\alpha}\}, \{B_{yb}\})$ be an arbitrary quantum strategy for~$\mathcal{G}_\mathrm{comp}$.
By \cref{eq:compiled winning prob} and the definition of the bias, it holds that
\begin{align*}
    \beta_q(S,\mathcal{G}_\mathrm{comp})&=\sum_{x,y}G_{xy}\mathop{\mathbb{E}}_{c\leftarrow\mathrm{Enc}(x)}\sum_{\alpha}(-1)^{\mathrm{Dec}(\alpha)}\bra{\psi_{c\alpha}}B_y\ket{\psi_{c\alpha}}\\
   &=\sum_x\mathop{\mathbb{E}}_{c\leftarrow\mathrm{Enc}(x)}\sum_{\alpha}(-1)^{\mathrm{Dec}(\alpha)}\bra{\psi_{c\alpha}} \left(\sum_yG_{xy}B_y\right)\ket{\psi_{c\alpha}}.
\end{align*}
By \cref{lemma:Jensen}, we get
\begin{align}\label{eq:after jensen}
    \beta_q(S,\mathcal{G}_\mathrm{comp})&\leq \sum_x\sqrt{\mathop{\mathbb{E}}_{c\leftarrow\mathrm{Enc}(x)}\sum_{\alpha}\bra{\psi_{c\alpha}} \left(\sum_yG_{xy}B_y\right)^2\ket{\psi_{c\alpha}}}.
    % \\
    % &=\sum_x\sqrt{\mathop{\mathbb{E}}_{c\leftarrow\mathrm{Enc}(x)}\sum_{\alpha}\bra{\psi_{c\alpha}} \sum_{s,y}G_{xs}G_{xy}B_sB_y\ket{\psi_{c\alpha}}}.
\end{align}
Now fix any~$x_0 \in I_A$.
For each~$x \in I_A$, the operator~$\mathcal B_x = \left(\sum_yG_{xy}B_y\right)^2$ has a QPT-implementable block encoding with scale factor~$\Theta(1)$ by \cref{lem:block from binary,lem:block lincom,lem:block prod}.
Moreover, $\norm{\mathcal B_x} = \norm{\sum_yG_{xy}B_y}^2 \leq \parens{ \sum_y \pi(x,y) }^2 \leq 1$.
Then, using \cref{lemma:Strongnonsignalling} with the distributions~$D_1=\delta_x$ and $D_2=\delta_{x_0}$ over Alice's questions, we see that there exists a negligible function~$\eta_x(\lambda)\geq0$ such that
\begin{align*}
  \left| \mathop{\mathbb{E}}_{c\leftarrow\mathrm{Enc}(x)}\sum_{\alpha}\bra{\psi_{c\alpha}} \mathcal B_x \ket{\psi_{c\alpha}}
  - \mathop{\mathbb{E}}_{c\leftarrow\mathrm{Enc}(x_0)}\sum_{\alpha}\bra{\psi_{c\alpha}} \mathcal B_x \ket{\psi_{c\alpha}} \right| \leq \eta_x(\lambda),
\end{align*}
and hence we can upper bound \cref{eq:after jensen} as follows:
\begin{align}
\nonumber
    \beta_q(S,\mathcal{G}_\mathrm{comp})&\leq
\sum_x\sqrt{\mathop{\mathbb{E}}_{c\leftarrow\mathrm{Enc}(x)}\sum_{\alpha}\bra{\psi_{c\alpha}} \mathcal B_x \ket{\psi_{c\alpha}}} \\
\nonumber
&= \sqrt{\mathop{\mathbb{E}}_{c\leftarrow\mathrm{Enc}(x_0)}\sum_{\alpha}\bra{\psi_{c\alpha}} \mathcal B_{x_0} \ket{\psi_{c\alpha}}} + \sum_{x\neq x_0}\sqrt{\mathop{\mathbb{E}}_{c\leftarrow\mathrm{Enc}(x)}\sum_{\alpha}\bra{\psi_{c\alpha}} \mathcal B_x \ket{\psi_{c\alpha}}} \\
\nonumber
&\leq \sqrt{\mathop{\mathbb{E}}_{c\leftarrow\mathrm{Enc}(x_0)}\sum_{\alpha}\bra{\psi_{c\alpha}} \mathcal B_{x_0} \ket{\psi_{c\alpha}}} + \sum_{x\neq x_0}\sqrt{\mathop{\mathbb{E}}_{c\leftarrow\mathrm{Enc}(x_0)}\sum_{\alpha}\bra{\psi_{c\alpha}} \mathcal B_x \ket{\psi_{c\alpha}} + \eta_x(\lambda)} \\
\nonumber
&\leq \sqrt{\mathop{\mathbb{E}}_{c\leftarrow\mathrm{Enc}(x_0)}\sum_{\alpha}\bra{\psi_{c\alpha}} \mathcal B_{x_0} \ket{\psi_{c\alpha}}} + \sum_{x\neq x_0}\sqrt{\mathop{\mathbb{E}}_{c\leftarrow\mathrm{Enc}(x_0)}\sum_{\alpha}\bra{\psi_{c\alpha}} \mathcal B_x \ket{\psi_{c\alpha}}} + \sum_{x \neq x_0} \sqrt{\eta_x(\lambda)} \\
\label{eq:proofXOR}
&= \sum_x \sqrt{\mathop{\mathbb{E}}_{c\leftarrow\mathrm{Enc}(x_0)}\sum_{\alpha}\bra{\psi_{c\alpha}} \mathcal B_x \ket{\psi_{c\alpha}}} + \eta'(\lambda),
% \nonumber
% &= \sum_x \sqrt{\sum_{s,y} G_{xs} G_{xy} \mathop{\mathbb{E}}_{c\leftarrow\mathrm{Enc}(x_0)}\sum_{\alpha}\bra{\psi_{c\alpha}} B_s B_y \ket{\psi_{c\alpha}}} + \eta'(\lambda)
\end{align}
where the last inequality is due to~$\sqrt{s+t} \leq \sqrt s + \sqrt t$, which holds for all~$s,t\geq0$, and where we introduced the negligible function~$\eta'(\lambda) := \sum_{x \neq x_0} \sqrt{\eta_x(\lambda)}$.
Finally, we observe that if we define a matrix~$V \in \R^{I_B \times I_B}$ with entries~$V_{sy} := \mathop{\mathbb{E}}_{c\leftarrow\mathrm{Enc}(x_0)}\sum_{\alpha}\frac{1}{2}\bra{\psi_{c\alpha}}B_s B_y+B_yB_s\ket{\psi_{c\alpha}}$, then there exists a vector strategy~$\{u_x,v_y\}$ such that, for all~$x\in I_A$,
\begin{align}\label{eq:proofXOR2}
  \sqrt{\mathop{\mathbb{E}}_{c\leftarrow\mathrm{Enc}(x_0)}\sum_{\alpha}\bra{\psi_{c\alpha}} \mathcal B_x \ket{\psi_{c\alpha}}}
= \sqrt{\sum_{s,y} G_{xs} G_{xy} V_{sy}}
= \sum_y G_{xy} \braket{u_x}{v_y}.
\end{align}
This follows from \cref{lemma:bias}, because~$V\in \R^{I_B \times I_B}$ is positive semidefinite (it is a nonnegative linear combination of Gram matrices) and also satisfies~$V_{yy} = 1$.
As a consequence of this and Tsirelson's theorem (\cref{thm:Tsirelson}), we obtain that
\begin{align}\label{eq:proofXOR3}
    \sum_x \sqrt{\mathop{\mathbb{E}}_{c\leftarrow\mathrm{Enc}(x_0)}\sum_{\alpha}\bra{\psi_{c\alpha}} \mathcal B_x \ket{\psi_{c\alpha}}}
= \sum_{x,y} G_{xy} \braket{u_x}{v_y}
\leq \beta_q^*(\mathcal{G}).
\end{align}
Thus we have proved that
\begin{align*}
    \beta_q(S,\mathcal{G}_\mathrm{comp}) \leq \beta_q^*(\mathcal{G}) + \eta'(\lambda).
\end{align*}
By definition of the biases, we get $\omega_q(S,\mathcal{G}_\mathrm{comp}) \leq \omega_q^*(\mathcal{G}) + \frac12\eta'(\lambda)$.
This finishes the proof.
\end{proof}

\end{document}